\documentclass[12pt, draftclsnofoot, onecolumn]{IEEEtran}
\ifCLASSINFOpdf
\else
\fi
\let\labelindent\relax
\usepackage{enumitem}
\usepackage{algpseudocode}
\usepackage{algorithm}
\usepackage{setspace}
\usepackage{graphicx}
\usepackage{notoccite}
\usepackage{amssymb}
\usepackage{amsfonts}
\usepackage{amsthm}
\usepackage{amsmath}
\usepackage{bbm}
\usepackage{subfigure}
\usepackage{multirow}
\usepackage{pdfrender}
\usepackage{color}

\newtheorem{Theorem}{Theorem}

\begin{document}
%
\title{Joint Caching and Resource Allocation in D2D-Assisted Wireless HetNet}
%
%
%

\author{Wael~Jaafar,~\IEEEmembership{Member,~IEEE,}
	Amina~Mseddi,~\IEEEmembership{Student Member,~IEEE,}
        Wessam~Ajib,~\IEEEmembership{Senior Member,~IEEE,}
        and~Halima~Elbiaze,~\IEEEmembership{Member,~IEEE}
\thanks{W. Jaafar is with Carleton University, ON, Canada. A. Mseddi, W. Ajib and H. Elbiaze are with University of Quebec in Montreal, QC, Canada. \protect \\ E-mail: waeljaafar@sce.carleton.ca \protect \\ E-mail: mseddi.amina@courrier.uqam.ca  \protect \\ E-mail: ajib.wessam@uqam.ca  \protect \\ E-mail: elbiaze.halima@uqam.ca }}


%
%

\markboth{Journal of \LaTeX\ Class Files,~Vol.~XX, No.~X, X~X}{Jaafar \MakeLowercase{\textit{et al.}}: Joint Caching and Resource Allocation in D2D-Assisted Wireless HetNet}
%



\maketitle

\begin{abstract}
5G networks are required to provide very fast and reliable communications while dealing with the increase of users traffic. In Heterogeneous Networks (HetNets) assisted with Device-to-Device (D2D) communication, traffic can be offloaded to Small Base Stations or to users to improve the network's successful data delivery rate. In this paper, we aim at maximizing the average number of files that are successfully delivered to users, by jointly optimizing caching placement and channel allocation in cache-enabled D2D-assisted HetNets. At first, an analytical upper-bound on the average content delivery delay is derived. Then, the joint optimization problem is formulated. The non-convexity of the problem is alleviated, and the optimal solution is determined. Due to the high time complexity of the obtained solution, a low-complex sub-optimal approach is proposed. Numerical results illustrate the efficacy of the proposed solutions and compare them to conventional approaches. Finally, by investigating the impact of key parameters, e.g. power, caching capacity, QoS requirements, etc., guidelines to design these networks are obtained. 
\end{abstract}


%
\IEEEpeerreviewmaketitle

\section{Introduction}
%
%
%
%
\IEEEPARstart{W}{ith} the fast growth of new mobile applications and increasing demand for heavy cellular traffic such as live streaming, language processing, augmented reality and face recognition, future wireless networks are required to focus not only on communication and connectivity, but also on high caching in order to satisfy such applications' demands \cite{Wang2018}. 

The proposal of Heterogeneous Networks (HetNets) using small cells and different radio technologies improves the Energy Efficiency (EE) and Spectral Efficiency (SE) performances \cite{Agiwal2016}. However, the ultra dense small cell deployment and coexistence/interaction with existing Macro Base Stations (MBS) infrastructure creates new problems of mutual interference, requiring efficient radio resource allocation and interference management techniques \cite{Cheng2016}.
Moreover, Device-to-Device (D2D) communication has been envisioned as an essential component of 5G wireless networks. By allowing two users (UEs) to communicate directly, without going through the cellular network, radio resources can be further exploited to enhance SE, transmission delay, and offload the backhaul's traffic \cite{Ansari}.  
            
On one hand, many works have proposed combining HetNets and D2D to leverage the advantages of D2D in the HetNet architecture. 
Malandrino \textit{et al.} \cite{Malan2015} solve the resource allocation problem in an LTE 2-tier HetNet, supporting in-band D2D communications. They proposed a dynamic programming approach to schedule the download and upload traffic. Ali \textit{et al.} \cite{Ali2016} studied the power allocation and cell selection problems to maximize EE for LTE-A HetNets supporting D2D and relaying. They showed that more UEs can perform in D2D communications with more stringent QoS requirements. 
Huang \textit{et al.} \cite{Huang2016} proposed a centrally controlled framework for a D2D communication underlaying a two-tier cellular network. The objective is to maximize the sum rate of the network with individual transmit power and rate constraints. The formulated problem encompasses cellular, dedicated and shared D2D modes, frequency sharing and power control. 
In \cite{Cao2017}, Cao \textit{et al.} designed an offloading scheme based on the usage of SBSs to offload traffic from MBS, and also to use some UEs as relays to offload traffic from SBSs for users not in the coverage area of a SBS, in  the context of massive machine-type communications. 
Similarly, Tsiropoulos \textit{et al.} \cite{Tsiropoulos2017} proposed a cooperation framework among heterogeneous wireless devices to improve spectrum access and system's capacity.  
In \cite{Naqvi2018}, Naqvi \textit{et al.} optimized EE of D2D communications in a hybrid cellular network composed of millimeter wave and microwave cells, using power control and radio resource allocation. Whereas,     
Hao \textit{et al.} \cite{Hao2018} investigated the EE-SE trade-off for D2D communications underlying HetNets. An effective two-stage solution is proposed to solve the power and spectrum allocation problems.

On the other hand, caching has attracted recent attention thanks to its ability to further reduce the backhaul traffic and eliminate duplicate transmissions of popular content \cite{Wang2018,Bastug2014}. Caching has been investigated for use in different network types \cite{Li2018}. 
In macro-cellular networks, caching is typically available only in the MBSs. Peng \textit{et al.} \cite{Peng2015} proposed content placement to minimize the average download delay of files by users. Blaszczyszyn \textit{et al.} \cite{Blasz2015} maximized the hitting probability when a UE can be served by several MBSs. They showed that optimal caching is related to the signal-to-interference-plus-noise ratio (SINR). Whereas, Kreishah \textit{et al.} \cite{Kreishah2015} investigated cooperative content placement methods among MBSs to minimize both download and caching costs.    
In HetNets, a part of popular content can be cached in the SBSs. 
Shanmugam \textit{et al.} \cite{Shanmugam2013} optimized cooperative content caching in a backhaul constrained HetNet, aiming at minimizing the content download delay. 
Guan \textit{et al.} \cite{Guan2014} designed a caching method for small cells, where users move frequently from a small cell to another. They assumed that users' trajectories are known at the SBSs. In \cite{Yang2014}, Yang \textit{et al.} integrated caching at the user side, where the latter can obtain the content from SBSs or from its cache memory. Energy savings are realized compared to non-caching at users. 
In D2D networks, the main objectives of caching are to improve the area spectral efficiency and provide low backhaul costs \cite{Asadi}. D2D networks leverage the caching memory of D2D devices to allow content delivery without going through base stations.   
Ji \textit{et al.} \cite{Ji2016} investigated the throughput-outage tradeoff of wireless networks, where clustered device caching via D2D communications is exploited. Zhang \textit{et al.}
\cite{Zhang2016} optimized D2D link scheduling and power allocation to maximize the system's throughput. To tackle the non-convexity of the original problem, a decomposition into a link scheduling and a power allocation problem is processed.  
Chen \textit{et al.} \cite{Chen2017} maximized the offloading gain of cache-enabled D2D networks by jointly optimizing caching and scheduling policies. 
In \cite{Yi2018}, Yi \textit{et al.} proposed a traffic offloading framework with social-aware D2D content sharing and caching. They formulated a cost maximization problem, parametrized by power control, channel allocation, link scheduling and reward design. The basis transformation method is introduced to solve the problem. Through theory and simulations, the superiority of this method in improving social welfare and network capacity is shown.  

Previous works investigated macro-cellular, HetNet, and D2D caching separately, or in a very limited collaborative way. Indeed, only recently caching optimization in multi-tier networks has attracted interest.  
Yang \textit{et al.} \cite{Yang2016} proposed and analyzed cache-based content delivery in a three-tier HetNet. While proactively caching popular content at SBSs and in a part of the UEs, they investigated theoretically the achievable maximum traffic load and global throughput gains. Results prove that the global throughput of cache-enabled systems can increase up to 57\% compared to systems without caching abilities. 
In \cite{Li2017}, Li \textit{et al.} presented a caching algorithm to minimize the average transmission delay in a macro-cell with D2D support. The proposed greedy algorithm performs better than popularity-based naive caching policy. 
Li \textit{et al.} \cite{Li2018_2} developed a general N-tier cache-enabled HetNet framework, where MBS, SBSs and pico BSs coexist. They proposed an optimal distributed caching scheme to maximize the successful delivery probability. They showed that the optimal solution depends on cache sizes and base stations densities. 
In \cite{Yang2018}, Yang \textit{et al.} proposed a similar framework to \cite{Li2018_2}. However, they aimed at minimizing the average file transmission delay through caching and bandwidth allocation. Their proposed low-complex sub-optimal solution achieved cache-buffer tradeoff, and is superior to conventional strategies not considering this tradeoff. 
Quer \textit{et al.} \cite{Quer2018} proposed a proactive caching strategy at both users and SBSs, where user mobility and different classes of user's interests for content are assumed, and the objective is to minimize the system's cost (e.g. energy, bandwidth). Optimal caching policy was obtained using standard integer programming optimization tools. 
In our work \cite{Jaafar2018}, we optimized caching and bandwidth allocation to minimize the average transmission delay of a D2D-assisted HetNet. We divided the problem into a bandwidth allocation and a caching problems. The first problem is optimally solved. Then, for a given bandwidth allocation strategy, two sub-optimal caching approaches are proposed. 
Finally, Amer \textit{et al.} \cite{Amer2018} designed a D2D caching framework with inter-cluster cooperation, where nodes of the same cluster cooperate through D2D communications, while nodes of different clusters exchange data through cellular transmissions. The formulated problem aims at minimizing the network average delay under caching, queuing and energy constraints. 

These works are promising, but they present some lacks. Indeed, \cite{Yang2016} results are limited, since the same copy of contents is cached in all nodes of the same-tier. This situation restrains the real potential of the multi-tier network. For \cite{Li2017}, no theoretical analysis was realized and optimization is limited to caching placement only, while \cite{Li2018_2,Yang2018} did not consider D2D communications. 
Also, \cite{Quer2018} focused on caching, and ignored the effects of wireless channels and transmit power conditions. In \cite{Jaafar2018}, only interference-free channels and probabilistic bandwidth allocation are assumed. Moreover, the formulated problem is not optimally solved. Finally, \cite{Amer2018} is limited to a two-tier network.



 

In an attempt to overcome the shortcomings of related works, we propose in this paper to investigate a multi-tier D2D-assisted HetNet, where caching at all tiers is considered. The main objective is to maximize the Successful Delivery Rate (SDR), under channel resources, caching, and delay constraints.

The contributions of this paper are summarized as follows:
\begin{enumerate}
	\item We analytically derive an upper-bound, based on Chernoff bound, on the average content delivery delay. The latter is used in the formulation of the joint caching and resource allocation optimization problem, aiming at maximizing the SDR of the system. Note that resource allocation in the formulated problem refers to channel allocation. The non-convex non-linear obtained problem is transformed into an Integer Linear Program (ILP), in order to achieve convexity. Then, the optimal solution is obtained using the IBM CPLEX solver.
	\item Due to the complexity in obtaining the optimal solution, a low-complex sub-optimal algorithm is proposed. First, the resource (i.e., channel) allocation and caching placement problems are separated. Then, based on a geometrical approach, we propose Algorithm \ref{Algo0} to efficiently allocate channels to users. For a given resource allocation strategy, we propose an iterative caching placement policy (detailed in Algorithm \ref{Algo00}) based on the files' popularity across the users' classes.
	\item Through numerical results, the superiority of the proposed solutions is illustrated compared to conventional approaches. Moreover, the impacts of key parameters, e.g. transmit power, cache size, packet length, number of users and SBSs, and delay requirements, are presented. From these results, guidelines to design D2D-assisted HetNets are deduced.  
\end{enumerate}

The rest of the paper is organized as follows. In Section II, the system model is presented. Section III derives the content delivery delay and formulates the optimization problem. Then, section IV details the proposed solutions. In Section V, numerical results are presented. Finally, Section VI closes the paper.

\begin{figure*}
	\centering
	\includegraphics[width=400pt]{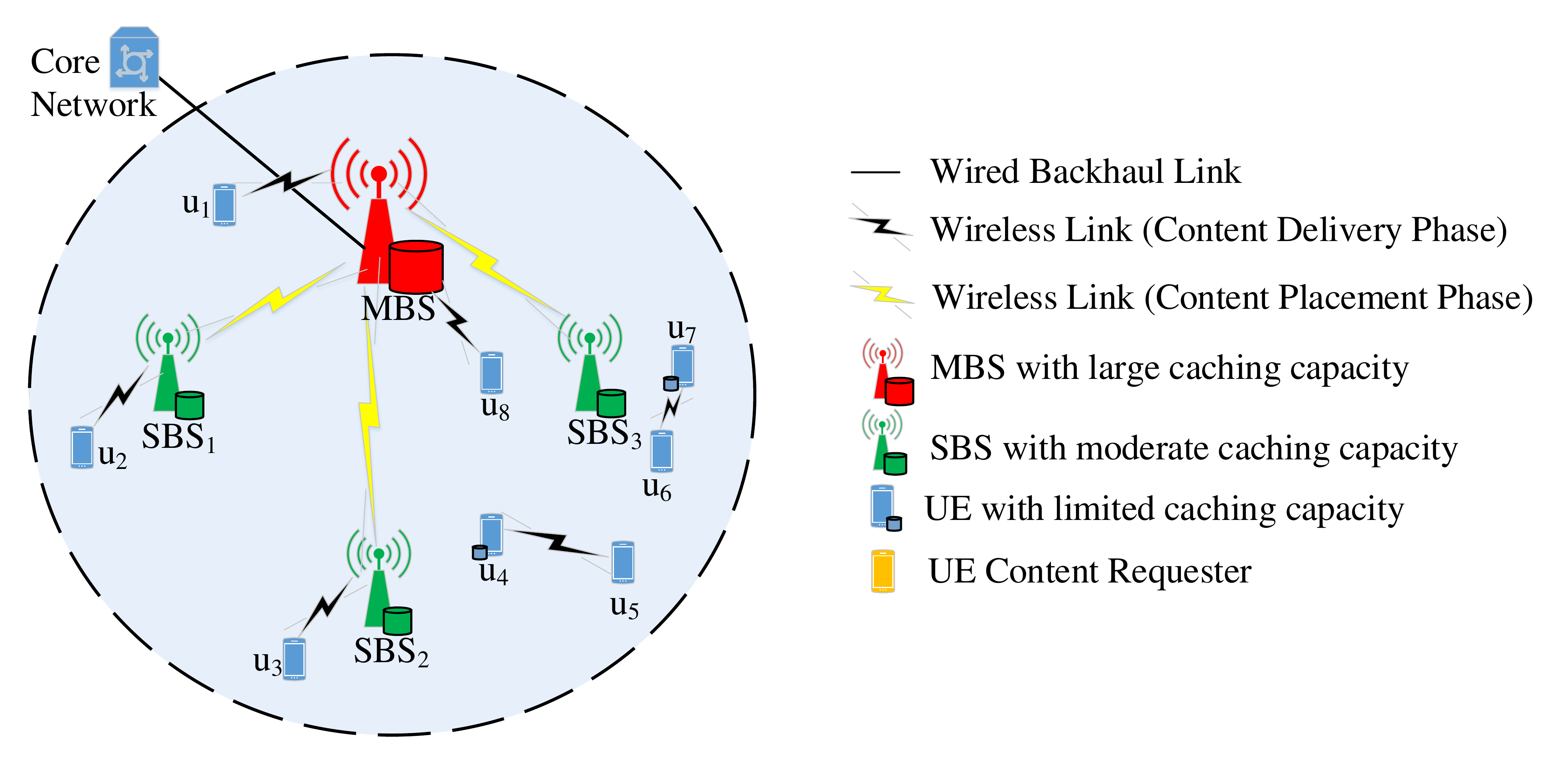}
	\caption{System Model}
	\label{Fig:System_Model}
\end{figure*}

\section{System Model}
We begin by modelling the network and wireless channels. Then, the caching model is presented, followed by the communication model. 

\subsection{Network and Channel Model}
The network consists of one MBS, $S$ SBSs and $U$ UEs, where SBS $s \in \mathcal{S}\equiv\{s_1,\ldots,s_S\}$ and UE $u \in \mathcal{U}\equiv\{u_1,\ldots,u_U\}$. Nodes of UEs-tier are deployed randomly within the MBS's coverage area as shown in Fig. \ref{Fig:System_Model}. In practice, the assumption that $U>S>1$ holds.

In HetNets, usually different MBSs share the frequency resource and SBSs reuse the same resources to save bandwidth. Moreover, D2D communication either uses dedicated cellular resources (Inband Overlay D2D) or reuses the same cellular resources as the base stations (Inband Underlay D2D)\cite{Asadi}. We assume that inter-cell interference is neglected thanks to control and scheduling among MBSs. However, we assume that intra-cell interference may occur when channels are reused at the same time.
Without loss of generality, we assume that $W$ orthogonal frequency channels of equal bandwidth $B$ are available in the network, where resource $w \in \mathcal{W}\equiv\{w_1,\ldots,w_W\}$. We define by $r_{u,w}$ the binary variable indicating whether channel $w$ is allocated to the transmission having user $u$ as its receiver or not, where $\sum_{w=1}^{W}r_{u,w}=1$. 

Moreover, we assume that time is divided into slots (TS), and in each TS $t \in \mathbb{N}$, transmissions may occur for one of the two following purposes: 1) off-peak content placement, where caching memories of the nodes are updated during low-activity periods within the network, and 2) content delivery to satisfy users' requests. This paper considers only transmissions related to content delivery. This is justified by the fact that off-peak content placement has no cost or impact on the resources' availability for content delivery.  

We assume that noise on wireless channels is independent identically distributed (i.i.d.) additive white Gaussian noise (AWGN) circularly symmetric with zero mean and variance $\sigma_0^2$, and that the channel coefficients are i.i.d.. We assume also that the channel coefficients are constant within a TS of duration $\tau$, and vary from a TS to another. 


\subsection{Caching Model}
In our network, we assume that $F$ files (or segments) of the same size $L$ bits belong to a library $\mathcal{F}\equiv\{1,\ldots,F\}$ in the core network,  \cite{Yang2016}. Without loss of generality, the same size of files can be obtained by splitting large files into equally sized segments. We assume that MBS, SBS $s$ and UE $u$ have caching capacities $C_m$ bits, $C_s$ bits and $C_u$ bits respectively, where $F> C_m > C_s> C_u$, $\forall s \in \mathcal{S}$ and $\forall u \in \mathcal{U}$.

We assume that user $u$ belongs to a specific \textit{class of interest} $k \in \mathcal{K}=\{1,\ldots,K\}$ that provides a specific ranking order of the file popularity. We introduce the probability that user $u$ belongs to a specific class $k$, defined as $p_{u}^k \in [0,1]$, where $\sum_{k=1}^{K}p_u^k=1$, $\forall u \in \mathcal{U}$. The definition of $p_{u}^k$ is motivated by the fact that the user's interests may vary in time due to its environment, behavior or events \cite{Quer2018}. 

For each class $k$, the file popularity follows Zipf distribution, and the probability that a user $u$ in class $k$ requests a file $f$ can be given by:
\begin{equation}
q_{f}^k=\frac{\eta(f,k)^{-\beta}}{\sum \limits_{f'=1}^F \eta(f',k)^{-\beta}}, \forall f \in \mathcal{F}, \forall k \in \mathcal{K},     
\end{equation}
where $\eta(f,k)$ (resp. $\eta(f',k)$) is the rank of file $f$ (resp. of file $f'$) for a user in class $k$, and $\beta \geq 0$ reflects how skewed the popularity distribution is, that means larger $\beta$ exponents correspond to higher content reuse, i.e., the first few popular files account for the majority of requests. 

In this paper, we assume that a file cannot be cached at more than one node within the system. This is motivated by the fact that redundant caching may lead to a fast saturation of storage resources, hence slowing down caching updates when old files need to be replaced by new ones. Moreover, caching at node $i \in \mathcal{C}=\mathcal{U}\cup \mathcal{S} \cup \{\text{MBS}\}$ is represented by a binary vector $\textbf{c}_i =\left[c_{i,1}, \ldots, c_{i,F}\right]$, where $c_{i,f}$ indicates whether node $i$ caches file $f$ or not \cite{Li2018}. 

\subsection{Communication Model}
When user $u$ requests file $f$, it starts by checking its own cache memory. If the content is available locally, it is obtained directly without any delay. Otherwise, the MBS will redirect the request to the most adequate source node. We assume that the MBS has knowledge of the nodes' caching status, as well as the statistics of the channel states within its cell. Moreover, we assume that: 1) A file transmission occupies one channel resource only; 2)
A content request is satisfied by one node only; 3)
A transmission can occupy several successive time slots; 4) The transmit power of MBS, SBS $s$ and user $u$ are fixed to $P_{m}$, $P_S$ and $P_U$ respectively.
Finally, only MBS and SBSs can communicate with several users simultaneously using orthogonal channels.

\section{Content Delivery Delay Expression and Problem Formulation}
In this section, we define the macro-cell's average content delivery delay and we derive its upper bound expression. Then, the latter is used in the formulation of the optimization problem.

\subsection{Content Delivery Delay and Upper Bound Derivation}
In this paper, the transmission delay between two nodes $i$ and $j$ is defined as the minimal number of time slots to transmit a given file $f$ of length $L$ from $i$ to $j$. 
For a time-varying channel, the average content delivery delay between node $i$ and node $j$ can be written as:
\begin{equation}
T_{ij}=\mathrm{min}\left\{T ; L \leq \sum \limits_{t=1}^{T}\tau R_{ij}(t) \right\},
\end{equation}
where $R_{ij}(t)$ is the channel's data rate in time slot $t$. Without loss of generality, for a given channel $w$, the data rate between node $i$ and node $j$ can be given by:
\begin{equation}
\label{eq:Rt}
R_{ij}(t)=B\; \mathrm{log}_2 \left(1+\frac{P_i^w ||h_{ij}(t)||^2}{ \sigma_0^2 + \sum_{i' \in \mathcal{I}_j}P_{i'}^w ||h_{i'j}(t)||^2 } \right)= B \; \mathrm{log}_2\left(1+\text{SINR}_{ij}(t)\right),
\end{equation} 
where $P_i^w$ is the transmit power of node $i \in \mathcal{C}$, $h_{ij}(t)=h'_{ij}(t){(\frac{d_{ij}}{d_0})}^{-\alpha}$ is the Rayleigh channel coefficient capturing both short-scale $h'_{ij}(t)$ and long-scale $\left(d_{ij}/d_0\right)^{-\alpha}$ fading, with $d_0$ is a reference distance, $||h_{ij}(t)||^2$ is the channel gain following an exponential distribution of zero mean and variance $(\frac{d_{ij}}{d_0})^{-\alpha}$, $\sum_{i' \in \mathcal{I}_j}P_{i'}^w ||h_{i'j}(t)||^2$ is the interference signal, and $\mathcal{I}_j$ is the set of interferers on transmission between $i$ and $j$.
The average content delivery delay of link $i-j$, denoted by $\bar{T}_{ij}$, can be defined as follows \cite{Peng2015}:  
\begin{equation}
\bar{T}_{ij}=\mathop{\mathbb{E}} \left[T_{ij} \right]=\sum \limits_{T=1}^{+\infty} \mathop{\mathbb{P}}\left[ T_{ij}> T\right]  
\end{equation}
where $\mathop{\mathbb{E}}$ is the expectation operator and $\mathop{\mathbb{P}}\left[ T_{ij}> T \right]$ is the probability that $T_{ij}$ is above $T$. 
This probability is determined by:
\begin{equation}
\label{eq:ProbT_ij}
\mathop{\mathbb{P}}\left[ T_{ij}>T\right]
= \mathop{\mathbb{P}}\left[ \sum \limits_{t=1}^{T} \mathrm{log}_2\left(1+\text{SINR}_{ij}(t)\right)< \frac{L}{\tau B}\right]=  \mathop{\mathbb{P}}\left[ \sum_{t=1}^T Y_t< \frac{L}{\tau B}\right]=\mathop{\mathbb{P}}\left[ Z_T< \frac{L}{\tau B}\right],
\end{equation}
where $Y_t=\mathrm{log}_2\left(1+\text{SINR}_{ij}(t)\right)$ and $Z_T=\sum_{t=1}^T Y_t$. 

\begin{Theorem}
	\label{Theorem0}
	$\mathop{\mathbb{P}}\left[ T_{ij}>T\right]$ is upper-bounded by:
	\begin{eqnarray}
	\label{eq:Chernoff1}
	\mathop{\mathbb{P}}\left[ T_{ij}>T\right]&\leq& \zeta_0(T):= \min_{t>0} \frac{e^{\frac{t L}{\tau B}+\frac{T}{\vartheta_{ij}} }}{\vartheta_{ij}^{\frac{t T}{\mathrm{ln}(2)}}}\left[\Gamma\left( 1-\frac{t}{\mathrm{ln}(2)},\frac{1}{\vartheta_{ij}} \right)\right]^T,
	\end{eqnarray}
	when $\mathcal{I}_j=\emptyset$, and by
	\begin{equation}
	\label{eq:Chernoff2}
	\mathop{\mathbb{P}}\left[ T_{ij}>T\right]\leq \zeta_1(T,\mathcal{I}_j):= \min_{t>0} \frac{\zeta_0(T)}{\vartheta_{ij}^{T \left(|\mathcal{I}_j|-1\right)}} \left(1+\sum \limits_{i'=1}^{|\mathcal{I}_j|}\vartheta_{i'j} \right)^T,
	\end{equation}
	when $\mathcal{I}_j \neq \emptyset$, where $\vartheta_{ij}=\frac{P_i^w }{\sigma_0^2}\cdot \left(\frac{d_{ij}}{d_0}\right)^{-\alpha}$, $\Gamma(a,x)=\int_{x}^{+\infty}s^{a-1}e^{-s} ds$ is the incomplete gamma function and $|.|$ is the cardinality operator.
\end{Theorem}
\begin{proof}
	Please refer to Appendix A.
\end{proof}
Consequently, the average content delivery delay of link $i-j$ can be upper bounded by:
\begin{equation}
\label{eq:ProbTij_UB}
\bar{T}_{ij}\leq   G_{ij}(\mathcal{I}_j):=\sum \limits_{T=1}^{+\infty} g_{ij}(T,\mathcal{I}_j),
\end{equation}
where $g_{ij}(T,\emptyset)=\zeta_0(T)$ and $g_{ij}(T,\mathcal{I}_j)=\zeta_1(T,\mathcal{I}_j)$, $\forall \; \mathcal{I}_j \neq \emptyset$.

We define by $\bar{D}_{u,f}$ the delay of delivering file $f$ to user $u$. According to the defined communication model, a file request is either served locally, or by another node. For the sake of delay expression $\bar{D}_{u,f}$ tractability and simplicity, we assume that when a channel is shared by simultaneous transmissions, interference is dominated by one component only. Hence, a communication is at most interfered by another transmission only.   
The previous assumption lead to the derivation of an upper bound on the average delay $\bar{D}_{u,f}$ given by:
	\begin{eqnarray}
	\label{eq:Dufk_LB}
	\bar{D}_{u,f}&\leq& G_{u,f}:=\left( 1-c_{u,f}\right)\nonumber \\ 
	&\cdot& \left[ \mathop{\mathbb{P}}\left[ \mathcal{T}_{u,f}=\emptyset\right] \left( \bar{T}_{0}+\mathop{\mathbb{P}}\left[ \mathcal{I}_{u}=\emptyset\right] {G}_{mu}(\emptyset) + \sum \limits_{\substack{y \in \mathcal{C}\backslash \{ u, m \}}  }\mathop{\mathbb{P}}\left[ \mathcal{I}_{u}=\{ y \}\right] {G}_{mu}(\mathcal{I}_u)\right)\right. \nonumber \\
	&+& \left. \sum \limits_{x \in \mathcal{C} \backslash\{u\}} \mathop{\mathbb{P}}\left[ \mathcal{T}_{u,f}=\{x\}\right] \left( \mathop{\mathbb{P}}\left[ \mathcal{I}_{u}=\emptyset\right] {G}_{xu}(\emptyset)+ \sum \limits_{\substack{y \in \mathcal{C} \backslash \{ u, x \}}} \mathop{\mathbb{P}}\left[ \mathcal{I}_{u}=\{y \}\right] {G}_{xu}(\mathcal{I}_u)  \right)\right],
	\end{eqnarray}
where $\mathcal{T}_{u,f}$ is the set of potential transmitters of $f$ to user $u$, $\bar{T}_{0}$ is the average transmission delay on the backhaul link, $\mathop{\mathbb{P}}\left[ \mathcal{T}_{u,f}=\emptyset\right]$ is the probability that no node is a potential transmitter of file $f$ to user $u$, $\mathop{\mathbb{P}}\left[ \mathcal{T}_{u,f}=\{x\}\right]$ is the probability that node $x$ is the potential transmitter of file $f$ to user $u$, $\mathop{\mathbb{P}}\left[ \mathcal{I}_{u}=\emptyset\right]$ is the probability that no node is interfering on user $u$'s communication, and $\mathop{\mathbb{P}}\left[ \mathcal{I}_{u}=\{y\}\right]$ is the probability that the transmission of node $y$ interferes on user $u$'s data reception.

$\mathop{\mathbb{P}}\left[ \mathcal{T}_{u,f}=\{x\}\right]$ can be expressed by:
\begin{equation}
\label{eq:Prob_T}
\mathop{\mathbb{P}}\left[ \mathcal{T}_{u,f}=\{x\}\right]= {c}_{x,f} \prod \limits_{x' \in \mathcal{C}\backslash\{u,x\}} \left( 1-{c}_{x',f} \right), 
\end{equation}
whereas $\mathop{\mathbb{P}}\left[ \mathcal{T}_{u,f}=\emptyset\right]=\prod \limits_{x \in \mathcal{C}\backslash\{u\}} \left( 1-{c}_{x,f} \right)$. In addition, $\mathop{\mathbb{P}}\left[ \mathcal{I}_{u}=\{y\}\right]$ can be given by:
\begin{eqnarray}
\label{eq:PSu}
\mathop{\mathbb{P}}\left[ \mathcal{I}_{u}=\{y\}\right]&=&\sum \limits_{w=1}^W  \sum \limits_{\substack{u' \in \mathcal{U} \\ \backslash \{u,x,y\}}} \Bigg(  r_{u,w} r_{u',w} \bigg( \sum \limits_{k'=1}^{K} \sum \limits_{\substack{f'=1\\ f' \neq f}}^F p_{u'}^{k'} q_{f'}^{k'} \mathop{\mathbb{P}}\left[\mathcal{T}_{u',f'}=\{y\} \right] \bigg) \Bigg), 
\end{eqnarray}
where the first two sums correspond to the channels allocated to the users and the next two sub-sums correspond to the class and the file requested by user $u' \neq u$. Finally, $\mathop{\mathbb{P}}\left[ \mathcal{I}_{u}=\emptyset\right]$ can be expressed by:
\begin{equation}
\label{eq:PS_emp}
\mathop{\mathbb{P}}\left[ \mathcal{I}_{u}=\emptyset\right]=\sum \limits_{w=1}^W  \left( r_{u,w}\prod \limits_{{u'\in \mathcal{U} \backslash\{u\}}}\left(1-r_{u',w}\right) \right).
\end{equation}


\subsection{Problem Formulation}
In this section, we formulate the joint caching and channel allocation optimization problem, aiming at maximizing SDR. 
The problem is formulated as follows (P1):
\begin{subequations}
\begin{align}
\max_{\mathbf{X},\mathbf{C},\mathbf{R}} & \quad 
O=\frac{1}{U}\sum \limits_{f=1}^F \sum_{u=1}^U \sum \limits_{k=1}^K p_u^k q_f^k x_{u,f}   \nonumber \\
\label{c1} 
\text{s.t.}\quad & D_{th}- G_{u,f}\geq G \left(x_{u,f}-1\right),\; \forall u \in \mathcal{U},\; \forall f \in \mathcal{F},   \tag{P1.a} \\
\label{c2}& \sum_{f=1}^{F} c_{i,f} \leq C'_i= C_i/L, \;\forall i \in \mathcal{C}, \tag{P1.b}\\
\label{c3}& \sum_{i=1}^{|\mathcal{C}|} c_{i,f} \leq 1, \;\forall f \in \mathcal{F}, \tag{P1.c}\\
\label{c4}&\sum_{w=1}^{W}r_{u,w}= 1, \;\forall u \in \mathcal{U}, \tag{P1.d}\\
\label{c5}& \sum_{u=1}^{U}r_{u,w}\leq R, \;\forall w \in \mathcal{W}, \tag{P1.e} \\
\label{c6}& x_{u,f} \in \{0,1\},\; \forall u \in \mathcal{U},\; \forall f \in \mathcal{F}, \tag{P1.f} \\
\label{c7}& c_{i,f} \in \{0,1\},\; \forall i \in \mathcal{C},\; \forall f \in \mathcal{F}, \tag{P1.g} \\
\label{c8}& r_{u,w}\in \{0,1\},\; \forall u \in \mathcal{U},\; \forall w \in \mathcal{W}, \tag{P1.h}
\end{align}
\end{subequations}
where $\textbf{X}=\left[ x_{u,f} \right]_{U \times F}$ is the matrix indicating the successful deliveries in the network, $\textbf{C}=\left[ c_{i,f} \right]_{|\mathcal{C}| \times F}$ is the caching placement matrix, and $\textbf{R}=\left[ r_{u,w} \right]_{U \times W}$ is the channel allocation matrix. 
The first constraint (\ref{c1}) emphasizes the fact that in average, the delivery delay of file $f$ to user $u$ cannot exceed a threshold $D_{th}$, where we opt for the use of the upper-bound $G_{u,f}$, and $G$ is a large constant guaranteeing that the constraint is valid at all times. Constraints (\ref{c2})-(\ref{c3}) express caching placement with memory limitation at node $i$ and caching redundancy limitation to $1$ within the macro-cell for any file $f$. Whereas, (\ref{c4})-(\ref{c5}) say that the same and unique channel is allocated for transmissions towards one user, and a channel $w$ can be allocated simultaneously to $R$ users at most. It is to be noted that $R$ is set to satisfy $R \geq \lceil{\frac{U}{W}}\rceil$ for fair channel allocation, where $\lceil{.}\rceil$ is the ceiling function. Finally, (\ref{c6})-(\ref{c8}) reflect the binary nature of the optimization variables. 

The defined optimization problem is non-linear due to the several product terms within the average content delivery delay expression in (\ref{c1}), and is non-convex due to the binary constraints. Nevertheless, in what follows, we propose a method to transform the problem and solve it optimally.

\section{Proposed Solutions}
We start in this section by proposing an approach to solve the problem optimally. Then, we propose two heuristic algorithms to find a sub-optimal solution with low complexity.
\subsection{Optimal Solution Design}
In order to make problem (P1) tractable, we propose to substitute iteratively the products of the binary variables in function $G_{u,f}$ by new single variables \cite{Quer2018,Crama2011}. 
 The resulting new problem (P2) can be expressed by:
\begin{subequations}
	\begin{align}
	\max_{\mathbf{X},\mathbf{C},\mathbf{R}} & \quad 
	O=\frac{1}{U}\sum \limits_{f=1}^F \sum_{u=1}^U \sum \limits_{k=1}^K p_u^k q_f^k x_{u,f}   \nonumber \\
	\text{s.t.}\quad & \text{(\ref{c1} - \ref{c8})}, \tag{P2.a - P2.h} \\
	\label{c21}
	&\Omega_{U,w,|\mathcal{C}|,f}^u \leq \varrho_{U,w}^u;\; \Omega_{U,w,|\mathcal{C}|,f}^u \leq \phi_{|\mathcal{C}|,f};\;\Omega_{U,w,|\mathcal{C}|,f}^u \geq \varrho_{U,w}^u+\phi_{|\mathcal{C}|,f}-1, \tag{P2.i} \\
		\label{c22}&  \Omega_{U,w,|\mathcal{C}|,f}^{u,x} \leq \varrho_{U,w}^u;\; \Omega_{U,w,|\mathcal{C}|,f}^{u,x} \leq \phi_{|\mathcal{C}|,f}^x;\;\Omega_{U,w,|\mathcal{C}|,f}^{u,x} \geq \varrho_{U,w}^u+\phi_{|\mathcal{C}|,f}^x-1,   \tag{P2.j}\\
	\label{c23}& \varrho_{u',w}^u \leq \varrho_{u'-1,w}^u;\; \varrho_{u',w}^u \leq \delta_{u',w}^u;\;\varrho_{u',w}^u \geq \varrho_{u'-1,w}^u+\delta_{u',w}^u-1,  \tag{P2.k}\\
	\label{c24}& \phi_{i,f}\leq \phi_{i-1,f};\; \phi_{i,f}\leq 1-c_{i,f};\; \phi_{i,f}\geq \phi_{i-1,f}-c_{i,f}, \tag{P2.l}\\
	\label{c25}& \phi_{i,f}^x\leq \phi_{i-1,f}^x;\; \phi_{i,f}^x\leq \bar{c}_{i,f}^x;\; \phi_{i,f}^x\geq \phi_{i-1,f}^x+\bar{c}_{i,f}^x-1, \tag{P2.m} \\
	\label{c26}& \Lambda_{u,u',w,|\mathcal{C}|,f,|\mathcal{C}|,f'}^y\leq \gamma_{u,u',w,|\mathcal{C}|,f};\; \Lambda_{u,u',w,|\mathcal{C}|,f,|\mathcal{C}|,f'}^y\leq \phi_{|\mathcal{C}|,f'}^y; \nonumber \\
	& \Lambda_{u,u',w,|\mathcal{C}|,f,|\mathcal{C}|,f'}^y\geq \gamma_{u,u',w,|\mathcal{C}|,f}+\phi_{|\mathcal{C}|,f'}^y-1, \tag{P2.n}  \\
			\label{c27}& \Lambda_{u,u',w,|\mathcal{C}|,f,|\mathcal{C}|,f'}^{x,y}\leq \gamma_{u,u',w,|\mathcal{C}|,f}^x;\; \Lambda_{u,u',w,|\mathcal{C}|,f,|\mathcal{C}|,f'}^{x,y}\leq \phi_{|\mathcal{C}|,f'}^y; \nonumber \\
	& \Lambda_{u,u',w,|\mathcal{C}|,f,|\mathcal{C}|,f'}^{x,y}\geq \gamma_{u,u',w,|\mathcal{C}|,f}^x+\phi_{|\mathcal{C}|,f'}^y-1, \tag{P2.o}  \\
	\label{c28}& \gamma_{u,u',w,|\mathcal{C}|,f}\leq r_{u,u',w};\;\gamma_{u,u',w,|\mathcal{C}|,f}\leq \phi_{|\mathcal{C}|,f'};\; \gamma_{u,u',w,|\mathcal{C}|,f}\geq r_{u,u',w}+\phi_{|\mathcal{C}|,f}-1, \tag{P2.p} \\
	\label{c29}& \gamma_{u,u',w,|\mathcal{C}|,f}^x\leq r_{u,u',w};\;\gamma_{u,u',w,|\mathcal{C}|,f}^x\leq \phi_{|\mathcal{C}|,f}^x;\; \gamma_{u,u',w,|\mathcal{C}|,f}^x\geq r_{u,u',w}+\phi_{|\mathcal{C}|,f}^x-1, \tag{P2.q} \\
		\label{c30}& r_{u,u',w}\leq r_{u,w};\;r_{u,u',w}\leq r_{u',w};\; r_{u,u',w}\geq r_{u,w}+r_{u',w}-1, \tag{P2.r}
	\end{align}
\end{subequations}
where in (\ref{c1})
\begin{eqnarray}
\label{eq:Guf_lin}
G_{u,f}&=& \bar{T}_0 \phi_{|\mathcal{C}|,f} + \sum_{w=1}^{W} \Omega_{U,w,|\mathcal{C}|,f}^u G_{mu}(\emptyset)+\sum_{w=1}^{W} \sum_{x \in \mathcal{C}\backslash\{u\}} \Omega_{U,w,|\mathcal{C}|,f}^{u,x} G_{xu}(\emptyset) \nonumber \\
&+& \sum_{w=1}^W \sum_{\substack{u'=1\\u' \neq u}}^U \sum_{k'=1}^{K} \sum_{\substack{f'=1;\\f' \neq f}}^F p_{u'}^{k'} q_{f'}^{k'} \sum_{\substack{y \in \mathcal{C}\\ \backslash\{u,u',m\}}} \Lambda_{u,u',w,|\mathcal{C}|,f,|\mathcal{C}|,f'}^y \nonumber \\
&+& \sum_{w=1}^W \sum_{\substack{u'=1\\u' \neq u}}^U \sum_{k'=1}^{K} \sum_{\substack{f'=1;\\f' \neq f}}^F p_{u'}^{k'} q_{f'}^{k'} \sum_{\substack{x \in \mathcal{C}\\ \backslash\{u,u'\}}} \sum_{\substack{y \in \mathcal{C}\\ \backslash\{u,u',x\}}} \Lambda_{u,u',w,|\mathcal{C}|,f,|\mathcal{C}|,f'}^{x,y},
\end{eqnarray}
and in (\ref{c21} - \ref{c30})
\begin{itemize}
	\item $\Omega_{U,w,|\mathcal{C}|,f}^u=\varrho_{U,w}^u \; \phi_{|\mathcal{C}|,f}$ and $\Omega_{U,w,|\mathcal{C}|,f}^{u,x}=\varrho_{U,w}^u\;  \phi_{|\mathcal{C}|,f}^x$,
	\item $\Lambda_{u,u',w,|\mathcal{C}|,f,|\mathcal{C}|,f'}^y=\gamma_{u,u',w,|\mathcal{C}|,f} \; \phi_{|\mathcal{C}|,f'}^y$ and $\Lambda_{u,u',w,|\mathcal{C}|,f,|\mathcal{C}|,f'}^{x,y}=\gamma_{u,u',w,|\mathcal{C}|,f}^x \; \phi_{|\mathcal{C}|,f'}^y$,
	\item $\varrho_{1,w}^u=\delta_{1,w}^u$ and $\varrho_{u',w}^u=\varrho_{u'-1,w}^u\delta_{u',w}^u$, $\forall u' \in \mathcal{U}\backslash\{1\}$ with $\delta_{u,w}^u=r_{u,w}$ and $\delta_{u',w}^u=1-r_{u',w}$, $\forall u' \in \mathcal{U}\backslash\{u\}$,
	\item $\phi_{1,f}=1-c_{1,f}$ and $\phi_{i,f}=\phi_{i-1,f} \left(1-c_{i,f}\right)$, $\forall i \in \mathcal{C}\backslash\{1\}$, 
	\item $\phi_{1,f}^x=\bar{c}_{1,f}^x$ and $\phi_{i,f}^x=\phi_{i-1,f}^x\; \bar{c}_{i,f}^x$, $\forall i \in \mathcal{C}\backslash\{1\}$, with $\bar{c}_{x,f}^x=c_{x,f}$ and $\bar{c}_{i,f}^x=1-c_{i,f}$, $\forall i \in \mathcal{C} \backslash\{x\}$,
	\item $\gamma_{u,u',w,|\mathcal{C}|,f}=r_{u,u',w} \;\phi_{|\mathcal{C}|,f}$ and $\gamma_{u,u',w,|\mathcal{C}|,f}^x=r_{u,u',w} \;\phi_{|\mathcal{C}|,f}^x$, with $r_{u,u',w}=r_{u,w}\;  r_{u'w}$.
\end{itemize}
Problem (P2) is now an ILP. We  model it using the AMPL language \cite{AMPL}, and then we solve it using IBM CPLEX Optimizer \cite{Cplex}.   

In order to quantify the complexity of problem (P2), we determine the total number of variables implicated in the optimization. For the variables in $\mathbf{X}$, the expected number of nonzero $x_{u,f}$ for each file $f$ can be defined by $N$. In (\ref{eq:Guf_lin}), the number of variables is dominant in the last term. Indeed, we can rewrite
\begin{equation}
\label{eq:Lambda_xy}
\Lambda_{u,u',w,|\mathcal{C}|,f,|\mathcal{C}|,f'}^{x,y}=r_{u,w}\; r_{u',w}\; \phi_{|\mathcal{C}|,f}^x  \; \phi_{|\mathcal{C}|,f'}^y, \; \forall u \in \mathcal{U}, u'\in \mathcal{U}\backslash\{u\},w \in \mathcal{W},f \in \mathcal{F},f' \in \mathcal{F}\backslash\{f\}.
\end{equation}     
Since the maximum number of nonzero $r_{u,w}$ and $r_{u',w}$ is limited by $R$, $\forall w \in \mathcal{W}$, and that of $\phi_{|\mathcal{C}|,f}^x$ (resp. $\phi_{|\mathcal{C}|,f'}^y$) is bounded by 1, $\forall f \in \mathcal{F}$ and $\forall x \in \mathcal{C}\backslash\{u,u'\}$ (resp. $\forall f' \in \mathcal{F}\backslash\{f\}$ and $\forall y \in \mathcal{C}\backslash\{u,u',x\}$), then the total number of variables in the optimization problem (P2) is of the order of
\begin{eqnarray}
\label{eq:complexity}
&&\mathcal{O}\left( \sum_{f=1}^F 2^{N} \times \sum_{w=1}^W \sum_{u=1}^U    \sum_{\substack{u'=1\\u' \neq u}}^U \sum_{f=1}^F \sum_{\substack{f'=1;\\f' \neq f}}^F  \sum_{\substack{x \in \mathcal{C}\\ \backslash\{u,u'\}}} \sum_{\substack{y \in \mathcal{C}\\ \backslash\{u,u',x\}}} 2^{2 \left(R +1\right)}\right) \nonumber\\
&=&\mathcal{O}\left( W U (U-1) F^2 (F-1) (|\mathcal{C}|-2)(|\mathcal{C}|-3) 2^{N+2 R +2} \right).
\end{eqnarray}

\subsection{Heuristic Solution Design}
Due to the high time complexity to obtain the optimal solution, we propose the following heuristic resource allocation and caching algorithm. We start by the resource (i.e. channel) allocation process, where channels are allocated in regards to minimize interference. Then, the caching process is performed, aiming at maximizing the number of popularity-wise user-file pairs respecting given constraints. 

\subsubsection{Resource Allocation}
Allocating channels to users depends on the available number of channels and on the number of users. Indeed, if $W \geq U$, it is optimal to dedicate one channel for every user. However, if $W < U$, channel allocation has to be adequately designed in order to minimize interference between simultaneous transmissions. 
 
To do so, we define the following parameters: let $X=\text{div}(U,W)$ and $R'=\text{mod}(U,W)$ be the quotient and remainder of the Euclidean division of $U$ by $W$ respectively. In order for the channels to be fairly used, we impose that for $W-R'$ channels, each channel has to be shared by $X$ users, while for the remaining $R'$ channels, each one is shared by $X+1$ users. Finally, we set $R=X+1=\lceil{\frac{U}{W}}\rceil$. 

Let $\mathcal{A}_{i}=\left\{ \mathcal{P}_1^i, \ldots, \mathcal{P}_{R'}^i, \bar{\mathcal{P}}_1^i, \ldots, \bar{\mathcal{P}}_{W-R'}^i \right\}$ be a set of $W$ polygons, such that each polygon $\mathcal{P}_j^i$ ($j \in \{1,\ldots, R'\}$) is a set of $X+1$ vertices (users) and each polygon $\bar{\mathcal{P}}_j^i$ ($j \in \{1,\ldots, W-R'\}$) is a set of $X$ vertices, where $\left\{\cup _{j=1}^{R'}\mathcal{P}_j^i \right\} \cup \left\{ \cup  _{j'=1}^{W-R'}\bar{\mathcal{P}}_j^i \right\} = \mathcal{U}$, and all polygons of $\mathcal{A}_i$ are distinct. 
 For each set $\mathcal{A}_i$, we calculate the perimeter $p_j^i$ of each polygon, the average perimeter $\bar{p}_i=\frac{1}{W}\sum_{j=1}^{W}p_j^i$ and the associated variance $\sigma_i^2=\frac{1}{W}\sum_{j=1}^{W}  \left( p_j^i-\bar{p}_i \right)^2$.
By defining the metric $\nu_i=\frac{\bar{p}_i}{\sigma_i^2}$, we allocate the $W$ channels to the users of polygons in set $\mathcal{A}_{i_0}=\text{arg}\max \limits_{i} \nu_i$. The key behind using the metric $\nu_i$ is that we aim at defining similar polygons of widely spaced users. The channel allocation procedure is summarized in Algorithm \ref{Algo0} below.

\begin{algorithm}
	\caption{Channel Allocation Algorithm}
	\label{Algo0}
	\begin{algorithmic}[1]
		\State Set $\mathbf{R}=\mathbf{0}_{(U \times W)}$ \% Channel Allocation Matrix
		\State Set $X=\text{div}(U,W)$; set $R'=\text{mod}(U,W)$ 
		\State Construct all sets $\mathcal{A}_i$ of $W$ polygons as described above
		\State Calculate $\bar{p}_i$, $\sigma_i^2$, and $\nu_i$, $\forall i$
		\State Find $\mathcal{A}_{i_0}=\text{arg}\max \limits_{i} \nu_i$
		\State Assign the $W$ channels to users in polygons of $\mathcal{A}_{i_0}$, by setting their binary variables in $\mathbf{R}$ to 1s
		\State Return $\mathbf{R}$
	\end{algorithmic}
\end{algorithm}

\subsubsection{Caching Placement}   
Conventional caching strategies would rely on popularity of files, where most popular files are cached closest to users. However, in our system a non-redundancy caching limit is imposed where a file cannot be cached more than once within the network. Consequently, the proposed caching strategy will place the files based in part on $Q_{u,f}=\sum _{k=1}^K p_u^k q_f^k$, the averaged file $f$ popularity for user $u$, with respect to the $K$ classes.   

For a given channel allocation, problem P1 can be decoupled into multiple cache placement problems (P$_f$) for $F$ files, such that:
\begin{subequations}
	\begin{align}
	\max_{\mathbf{X},\mathbf{C}} & \quad 
	O_f=\frac{1}{U} \sum_{u=1}^U Q_{u,f} x_{u,f}   \nonumber \\
	\label{c1_1} 
	\text{s.t.}\quad & (\ref{c1})-(\ref{c3}), (\ref{c6})-(\ref{c7}),
	\end{align}
\end{subequations}
We define by $Q_f=\sum_{u=1}^U Q_{u,f}$ the popularity metric of file $f$. By ranking $Q_f$ from the largest to the smallest value in a set $\mathcal{Q}$, we can proceed to an iterative file placement approach as follows: (1) Starting from the first file in $\mathcal{Q}$, we find its best location within the network that achieves the maximum $O_f$ value, with respect to constraints. (2) We update the caching memories. (3) We move to the next file in $\mathcal{Q}$ and repeat the first two steps (1)-(2). (4) We repeat step (3) until all files are processed. This approach is summarized in Algorithm \ref{Algo00} below.

\begin{algorithm}
	\caption{Caching Policy Algorithm}
	\label{Algo00}
\begin{algorithmic}[1]
	\State Set $\mathbf{C}=[c_{i,f}]=\mathbf{0}_{(|\mathcal{C}| \times F)}$ \% Caching Matrix
	\State $\forall f \in \mathcal{F}$, calculate $Q_{f}$
	\State Rank $Q_{f}$ from the largest to the smallest value in set $\mathcal{Q}$
	\For {$Q_{f'} \in \mathcal{Q}$}
	\State Set $\mathcal{C}'=\mathcal{C}$
	\State Calculate $O_{f'}$ when $f'$ is placed in node $i$, $\forall i \in \mathcal{C}'$
	\State Find $i_0=$arg$\max \limits_{i \in \mathcal{C'}} O_{f'}$ with respect to (\ref{c1})-(\ref{c3}), (\ref{c6})-(\ref{c7})
	\If{$\sum_{f=1}^{F}c_{i_0,f}<C'_i$}
	\State Set $c_{i_0,f'}=1$
	\Else
	\State $\mathcal{C}'=\mathcal{C}'\backslash\{i_0\}$
	\State Return to step 6
	\EndIf
	\EndFor
	\State Return $\mathbf{C}$
\end{algorithmic}
\end{algorithm}

\section{Numerical Results}
In this section, numerical results emphasize the performances of the proposed algorithms. Simulations assume that the system is composed of one MBS that occupies the center of a circular region, with radius equal to 100 meters, four uniformly located SBSs within $10\sqrt{50}\approx 71$ meters from the MBS, and $U \in \{10, 22\}$ randomly located users within the MBS's region. When not indicated, the following parameters are assumed fixed along the section: 
\begin{enumerate}
    \item Channel parameters: we fix the number of channels to $W=3$, with a bandwidth $B=1$ MHz for each channel. We assume that the path-loss exponent is set to $\alpha=3$, the time slot duration is $\tau=0.01$ sec, the average backhaul delay $\bar{T}_{0}=10$ TS (i.e. $0.1$ sec) and the transmission's success threshold $D_{th}=5$ TS.
    \item Caching parameters: We assume that the library has $F=100$ files, each with length $L=100$ bits. Moreover, the MBS, SBS and UE caching capacities are 500, 200 and 100 bits respectively. The Zipf distribution popularity exponent is set to $\beta=2$. Users can belong to one of the $K=3$ classes of interest, where the user $u$'s vector of probabilities to belong to a given class is defined as:
 \begin{equation}
    \label{eq:Pcla}
       \left[ p_u^1, p_u^2, p_u^3 \right]=\left\{
    \begin{array}{ll}
        \left[ 0.3, 0.5, 0.2 \right], & \mbox{if } \text{mod}(u,K)=0, \\
        \left[ 0.2, 0.3, 0.5 \right], & \mbox{if } \text{mod}(u,K)=1,  \\
        \left[ 0.5, 0.2, 0.3 \right], & \mbox{if } \text{mod}(u,K)=2, 
    \end{array}
\right.
\end{equation}   
\item Communication parameters: The transmit powers of MBS, SBSs and UEs in a given channel $w$ are set to $P_m=1$ Watt, $P_S=0.5$ Watts, and $P_U=0.1$ Watts respectively, while the noise power is fixed to $0.01$ Watts.
\end{enumerate}

\begin{figure}
	\begin{tabular}{cc}
		\centering
		\includegraphics[width=240pt,height=200pt]{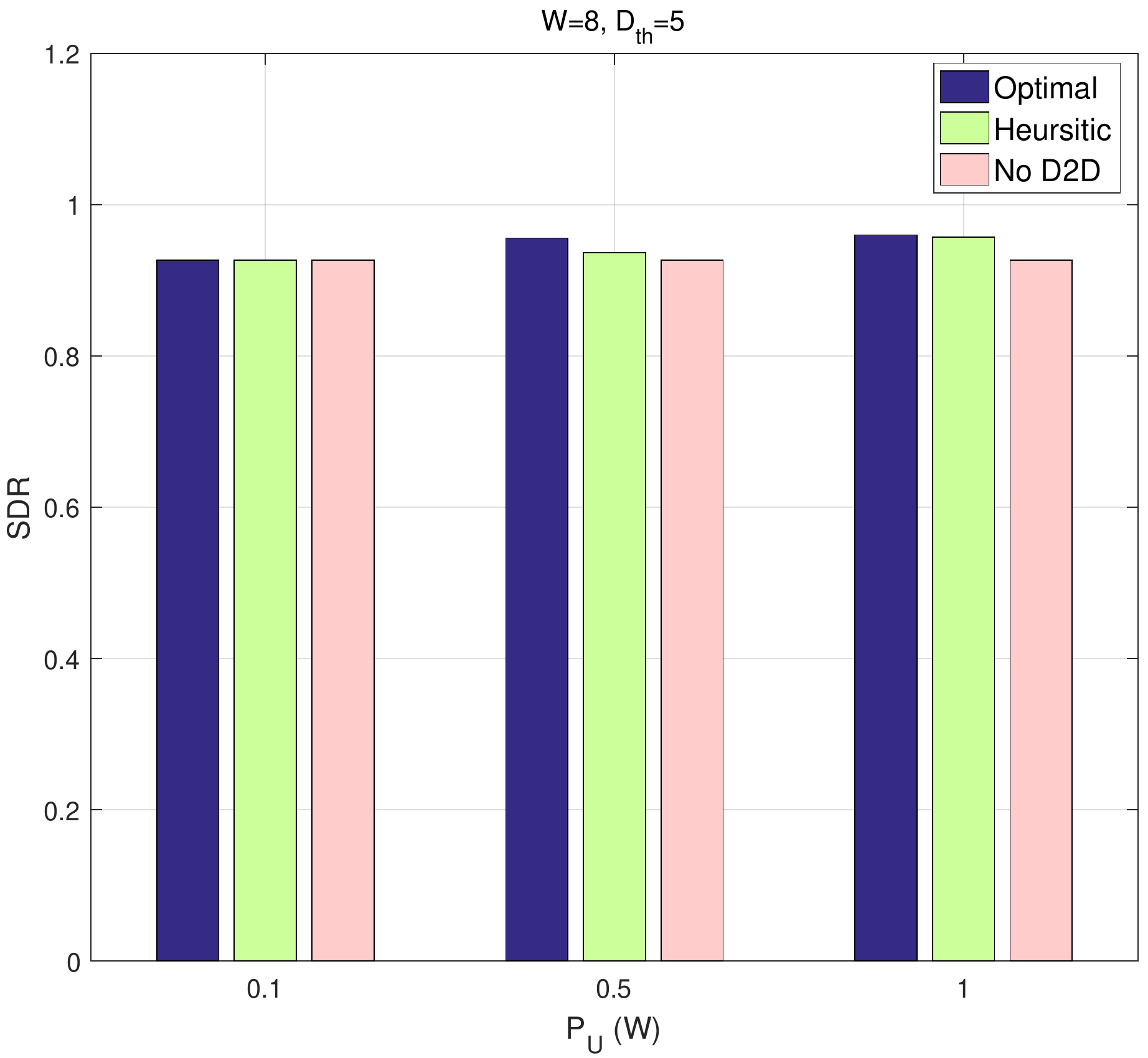} &   \includegraphics[width=240pt,height=200pt]{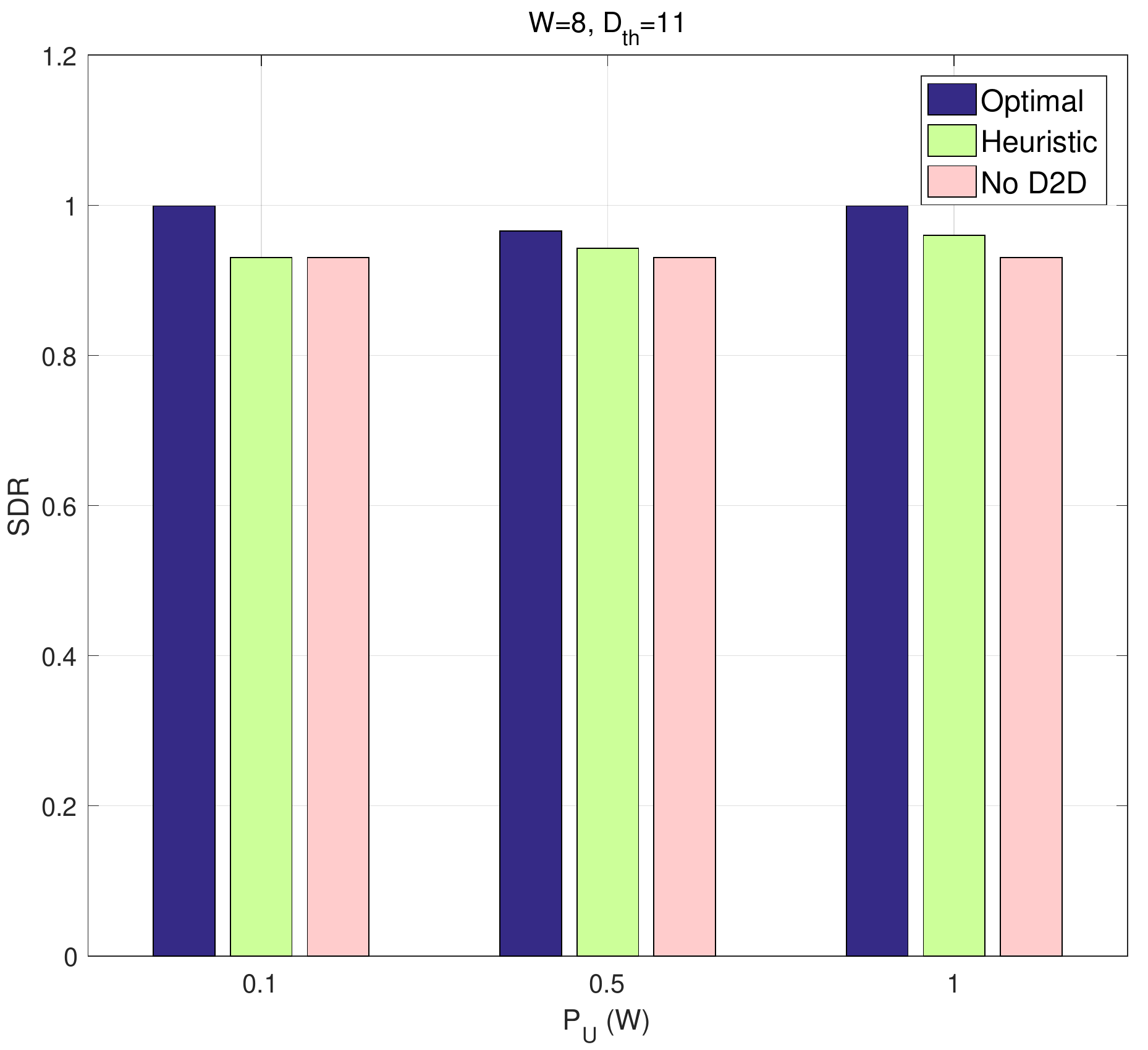} \\
		(a) SDR vs. $P_U$ ($W=8$ and $D_{th}=5$) & (b) SDR vs. $P_U$ ($W=8$ and $D_{th}=11$) \\[6pt]
		\includegraphics[width=240pt,height=200pt]{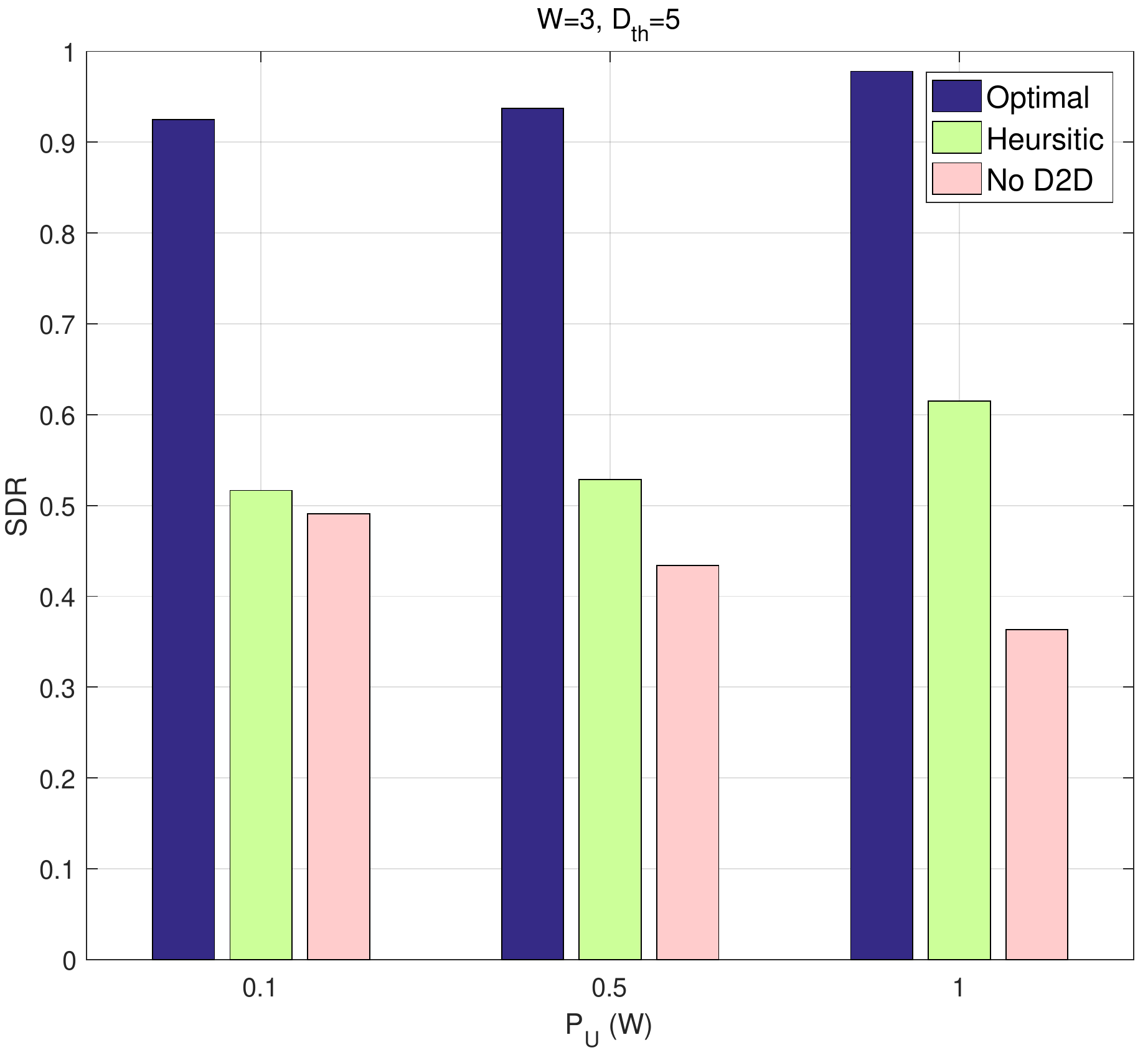} &   \includegraphics[width=240pt,height=200pt]{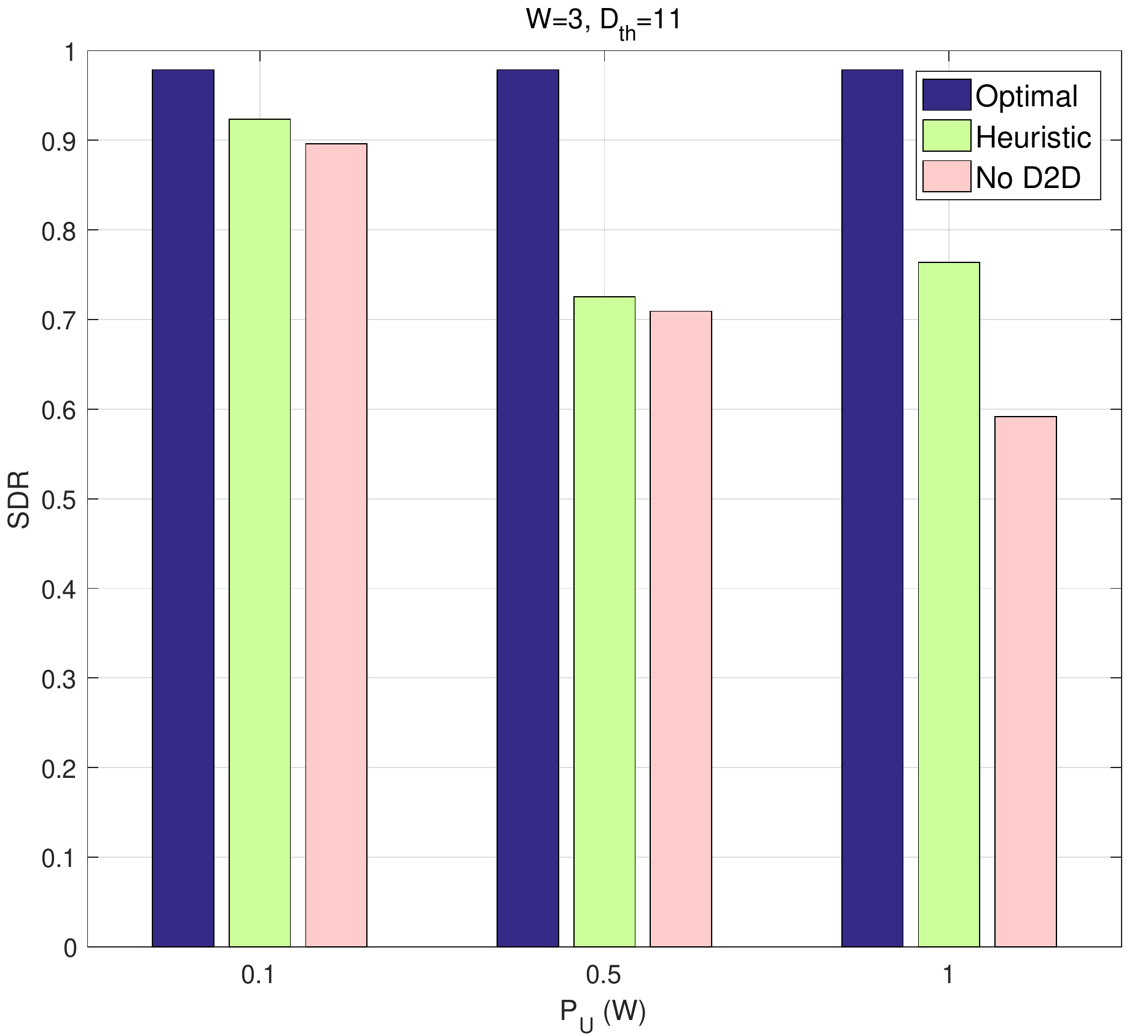} \\
		(c) SDR vs. $P_U$ ($W=3$ and $D_{th}=5$) & (d) SDR vs. $P_U$ ($W=3$ and $D_{th}=11$) \\[6pt]
	\end{tabular}
	\caption{Comparison of ``Optimal'', ``Heuristic'' and ``No D2D'' for different $P_U$ values}
	\label{fig:1}
\end{figure}

In Fig. \ref{fig:1} we present the SDR performance as a function of $P_U$ for the optimal solution (``Optimal''), the proposed heuristic (``Heuristic''), and the conventional Heterogeneous network (without D2D) (``No D2D'') \cite{Yang2018}. 
The use-case is composed of a small macro-cell with only $U=8$ users and $F=10$ files. The choice of such small parameters is due to the combination blast problem encountered for larger parameters' values. 
Fig. \ref{fig:1} compares these approaches for different $D_{th}$ and $W$ values. The optimal solution achieves performances that are close to 1 in most cases. Next, ``Heuristic'' achieves higher performances than ``No D2D''. For $W=8$ (Figs. \ref{fig:1}(a)-\ref{fig:1}(b)), communications are orthogonal, and hence, as $P_U$ increases, SDR for ``Heuristic'' improves. Indeed, with higher power, users become interesting candidates to cache popular files. However, SDR for ``No D2D'' is stable since it does not leverage caching at users. In fact, files are stored either in the core network or at the BSs. 

\begin{table}
	\centering
		\caption{Execution times of algorithms ($P_U=0.1$ W)}
	\label{table:1}
	\begin{tabular}{ |c|c|c|c|c| } 
		\hline
		{Algorithm} & $W=8$, $T_{th}=5$ & $W=8$, $T_{th}=11$ & $W=3$, $T_{th}=5$ & $W=3$, $T_{th}=11$ \\ 
		\hline \hline
		Optimal &483.59 sec & 518.38 sec & 74238 sec & 1025.93 sec \\ 
		\hline 
		Heuristic &0.00629 sec & 0.0063 sec & 0.00562 sec & 0.00664 sec \\
		\hline 
		No D2D &0.00266 sec & 0.00255 sec & 0.024 sec & 0.0105 sec \\
		\hline
	\end{tabular}
\end{table}

For $W=3$ (Figs. \ref{fig:1}(c)-\ref{fig:1}(d)), due to interference among communications, SDR degrades for ``Heuristic'' and ``No D2D''. This degradation is significant for stringent QoS ($D_{th}=5$). Indeed, this case demonstrates the limits of the proposed heuristic approach, as it leads to sub-optimal performances. Nevertheless, the complexity is dramatically reduced compared to ``Optimal''. This can be seen in Table \ref{table:1} where the execution time of ``Heuristic'' is in the order of milliseconds, while ``Optimal'' needs hundreds-thousands seconds to achieve its performance.   
Finally, in Fig. \ref{fig:1}(d), we notice that as $P_U$ increases, SDR is worst for $P_U=0.5$. Indeed, when $P_U$ is low, D2D communications are favored and they do not affect significantly simultaneous transmissions in the same frequency band. However, when $P_U$ reaches 0.5, interference becomes more important and inevitable when delivering files. Consequently, SDR degrades. However, as $P_U$ increases to 1, the approach results in a better caching strategy with less interfering D2D communications and thus improves SDR. According to this result, power control is required when conceiving D2D-assisted HetNets. The design of power control mechanisms is left for future works.

\begin{figure}
	\centering
	\includegraphics[width=350pt]{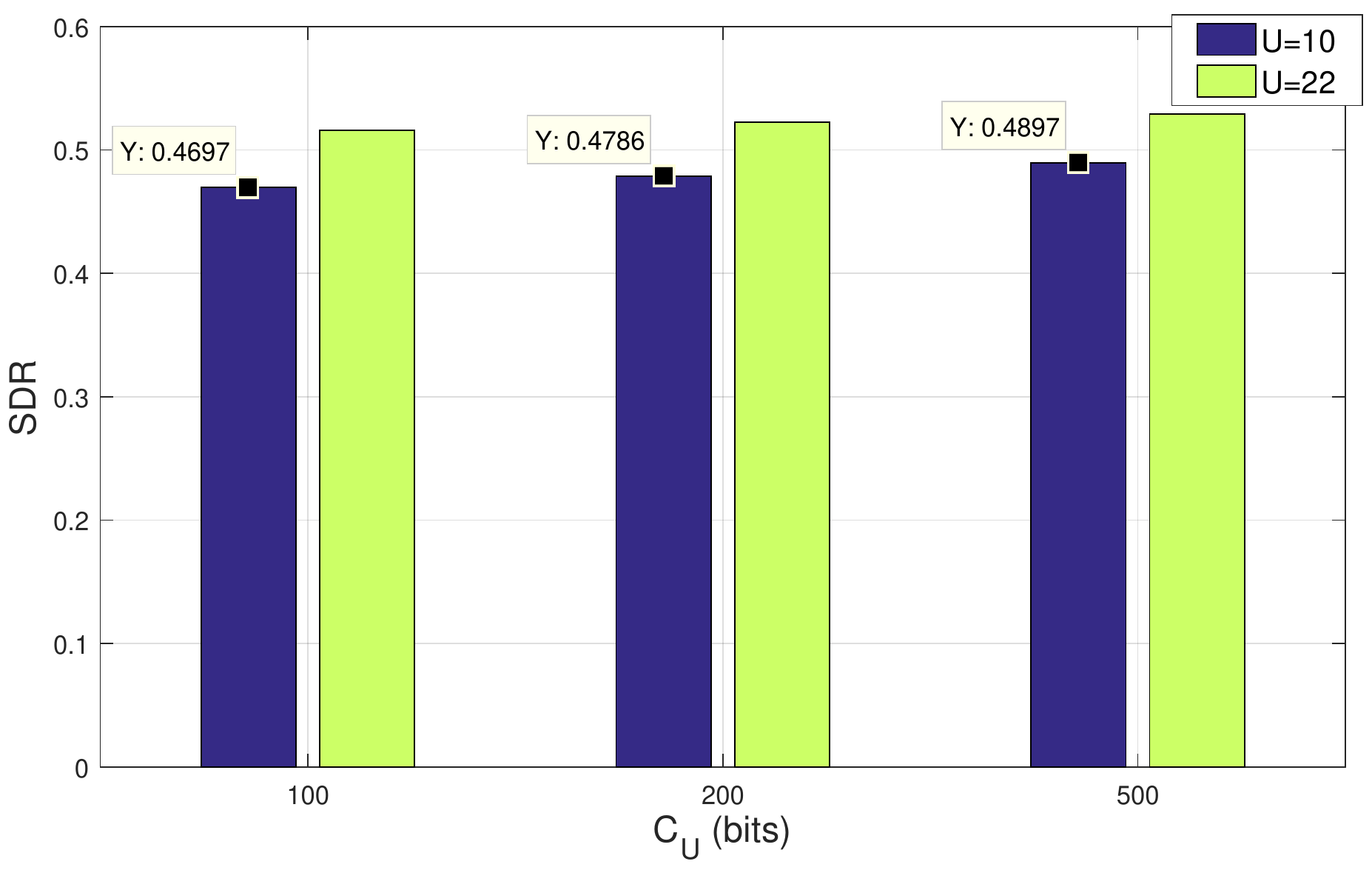}
	\caption{Influence of number of users $U$ and caching capacity $C_U$ on SDR}
	\label{Fig03}
\end{figure}

In order to evaluate larger sizes of the system, the remaining of this section shows only the performance of the ``Heuristic" algorithm in an interfered environment. We assume in what follows that $U=22$ users and $W=3$ channels. 

Fig. \ref{Fig03} illustrates SDR as a function of $C_U$, for $U=10$ and $22$ users respectively. When the caching capacity of users $C_U$ increases, the SDR performance enhances. Indeed, higher $C_U$ favors more caching in users and D2D low-interfered communications among them. Moreover, as $U$ grows, SDR significantly improves. This is due to a better caching policy that places files within best located users, compared to the case with a lower number of users. Hence, incentivizing users to participate, by caching and delivering files to other users, is advantageous in terms of SDR. 

\begin{figure}
	\centering
	\includegraphics[width=350pt]{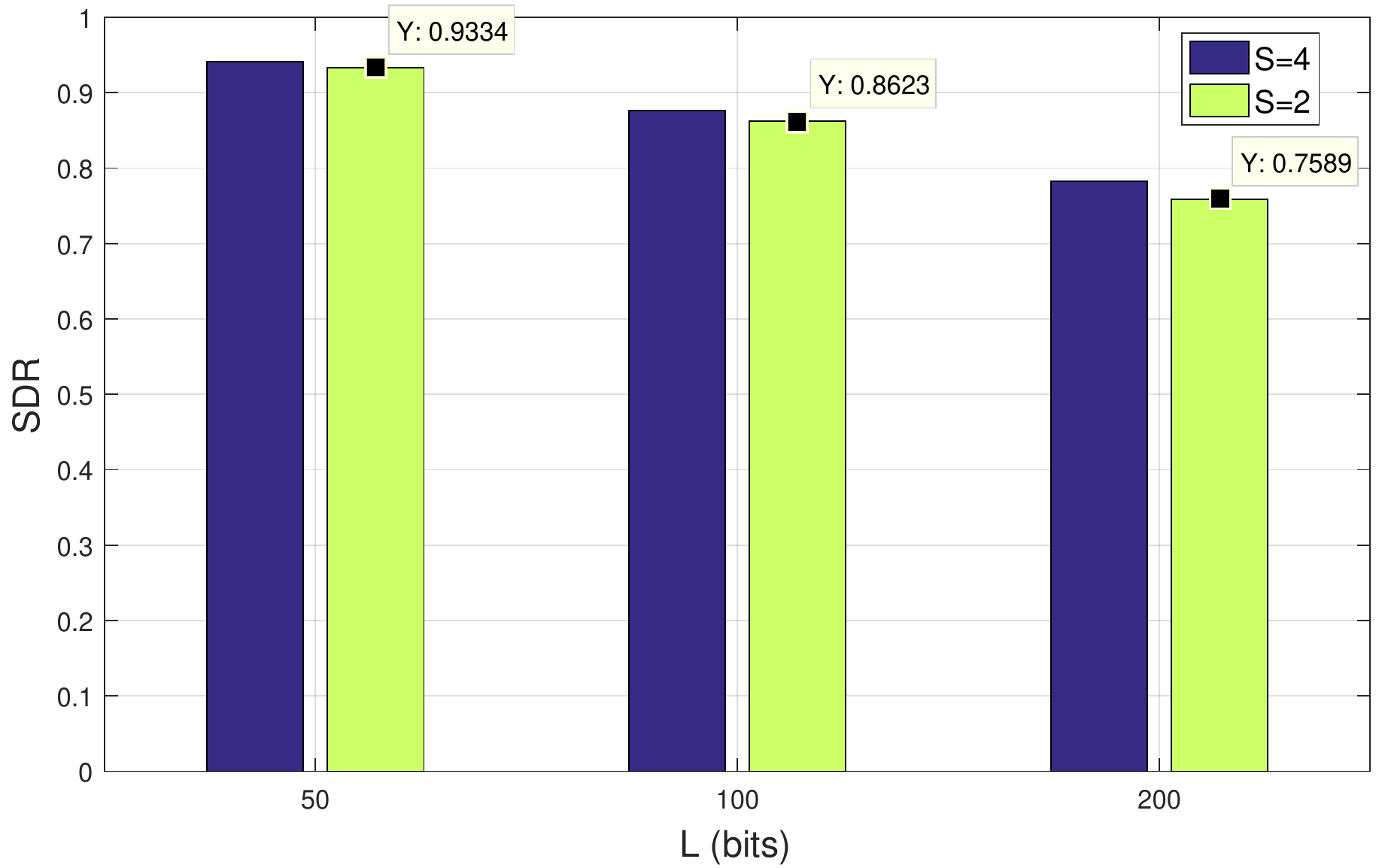}
	\caption{Influence of length of files $L$ and number of SBSs $S$ on SDR}
	\label{Fig04}
\end{figure}

Fig. \ref{Fig04} investigates the impact of the file length $L$ on SDR for $S=2$ and $4$ SBSs respectively. With larger files (or segments), caching within the network saturates rapidly and most files or segments have to be delivered via the backhaul link. Thus, due to a poor backhaul link and/or to a stringent delay threshold, SDR significantly degrades. In addition, SDR slightly decreases with the number of SBSs. Indeed, a small $S$ means that less caching capacity and reliable wireless links are available to deliver files to users within the macro-cell. Consequently, network densifying provides SDR gain if the additional interference does not affect content reception.

\begin{figure}
	\centering
	\includegraphics[width=350pt]{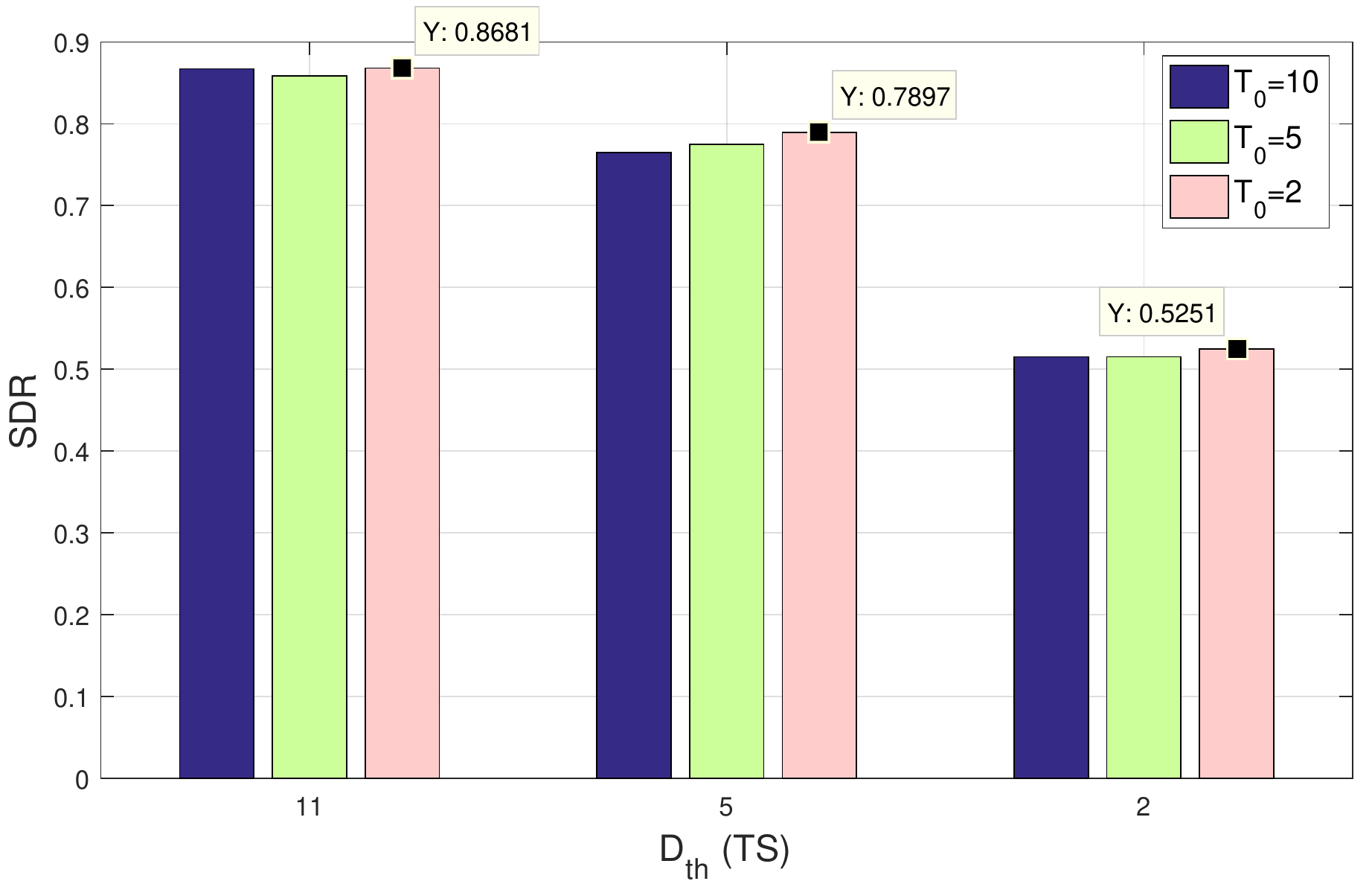}
	\caption{Influence of delay threshold $D_{th}$ and backhaul delay $\bar{T}_0$ on SDR}
	\label{Fig05}
\end{figure}

In Fig. \ref{Fig05}, the impact of the threshold $D_{th}$ on SDR is studied for $\bar{T}_{0} \in \{11\; \text{TS},5\; \text{TS},2\; \text{TS}\}$ respectively. On one hand, when $D_{th}$ becomes more stringent, SDR performances decreases. This is expected since a stricter delay threshold would inevitably penalize more files' deliveries. On the other hand, a better backhaul link is expected to improve SDR. However, as it can be seen for $D_{th}=11$, it is not always the case. Indeed, even with a better backhaul link, the caching policy may not ask the core network for most popular files that account for the larger part of users' requests. As a consequence, the backhaul's contribution in SDR gain is either low (for small $D_{th}$) or insignificant (for large $D_{th}$).  


\section{Conclusion}
In this paper, we studied joint caching and resource allocation for D2D-assisted wireless HetNets. We started by deriving an upper-bound of the average content delivery delay. The latter is used in formulating the maximization problem of SDR performance under communication and caching constraints. The joint caching and channel allocation problem is identified as non-linear and non-convex. To solve it optimally, we propose linearizion transformation into an ILP problem. Due to its complexity, the latter is solved optimally for small scenarios using CPLEX. Then, a low-complex two-step heuristic solution is proposed. In the first step, the channel allocation strategy is defined, while in the second step, caching policy is determined. Through numerical results, we have shown that the proposed heuristic algorithm achieves 60\%-90\% of the optimal solution's performance, with a very small execution time. Besides, it outperforms conventional caching approaches. Finally, the impact of key parameters is investigated. The obtained results provide the following design guidelines: 1) Incentivizing users to participate in the caching process, by making available large caching capacities, may be beneficial to both the system and users. 2) Densifying the network to some extent may bring some gain. However, this is true as long as the generated additional interference of new SBSs does not degrade significantly communications among the network's nodes. 3) A better backhaul link is generally advantageous when stringent QoS requirements are setup. However, for loosened QoS requirements, the system may decrease its backhaul use by caching more files within the system's nodes. As a future work, the system may be extended to heterogeneous QoS requirements, caching capacities, transmit powers and supported number of channels.  

\appendices
\section{Proof of Theorem \ref{Theorem0}}
\label{FirstAppendix}
By applying the Chernoff bound to (\ref{eq:ProbT_ij}) \cite{mitzenmacher_upfal_2005}, and assuming independence between random variables (RV) $Y_n$, $\forall n=1,\ldots,T$, we obtain:
\begin{equation}
\label{eq:develop1}
\mathop{\mathbb{P}}\left[ Z_T< \frac{L}{\tau B}\right]\leq \min_{t>0} e^{\frac{t L}{\tau B}} \prod \limits_{n=1}^T \mathop{\mathbb{E}}\left[e^{-t Y_n}\right].
\end{equation} 
Let $Z_n=e^{-t Y_n}$ and $X_n=\mathrm{SINR_{ij}(n)}$, $\forall n=1,\ldots,T$, then the cumulative distribution function (cdf) of $Z_n$ can be expressed by ($\forall z \in (0,1]$):
\begin{equation}
\label{eq:developp2}
F_{Z_n}(z)=\mathop{\mathbb{P}}\left[ e^{-t\; \mathrm{log}_2(1+X_n)}< z\right]=\mathop{\mathbb{P}}\left[ X_n\geq 2^{-\frac{\mathrm{ln}(z)}{t}-1}\right] =1-F_{X_n}\left(2^{-\frac{\mathrm{ln}(z)}{t}-1}\right)=e^{-\frac{2^{-\frac{\mathrm{ln}(z)}{t}-1}}{\vartheta_{ij}}},
\end{equation} 
where $F_{X_n}(x)=1-e^{-\frac{x}{\vartheta_{ij}}}$, $\forall x○ \geq 0$, is the cdf of RV $X_n$ in an interference-free environment. From (\ref{eq:developp2}), we derive the probability density function (pdf) of $Z_n$, given by:
\begin{equation}
\label{eq:develop3}
f_{Z_n}(z)=\frac{\partial F_{Z_n}(z)}{\partial z} =\frac{\mathrm{ln}(2) e^{-\frac{2^{-\frac{\mathrm{ln}(z)}{t}}-1}{\vartheta_{ij}}}}{z \; t \; \vartheta_{ij} 2^{\frac{\mathrm{ln}(z)}{t}}}, \; \forall \; 0<z\leq 1.
\end{equation}   
Consequently, the mean of $Z_n$ can be calculated as follows:
\begin{equation}
\label{eq:develop4}
\mathop{\mathbb{E}}\left[Z_n\right]=\int \limits_{0}^{1} z f_{Z_n}(z)dz= \frac{e^{\frac{1}{\vartheta_{ij}}} \Gamma\left(1-\frac{t}{\mathrm{ln}(2)},\frac{1}{\vartheta_{ij} }\right) }{\vartheta_{ij}^{\frac{t}{\mathrm{ln}(2)}}}.
\end{equation}
By combining (\ref{eq:develop4}) into (\ref{eq:develop1}), $\zeta_0(T)$ is obtained, given by (\ref{eq:Chernoff1}). 

In an interfered environment, the cdf of $X_n$ can be obtained from (Eq. A.4, \cite{Yang2018}) as:
\begin{equation}
\label{eq:develop5}
F_{X_n}(x)=1-\frac{\vartheta_{ij} e^{-\frac{x}{\vartheta_{ij}}}}{\prod \limits_{i'=1}^{|\mathcal{I}_j|} \left(\vartheta_{ij}+\vartheta_{i'j} \; x\right)}, \; \forall \; x \geq 0. 
\end{equation}
Hence, the cdf of $Z_n$ can be obtained similarly to (\ref{eq:developp2}), and its pdf is given as follows:
\begin{equation}
\label{eq:develop6}
f_{Z_n}^0(z)=\frac{\left(1+\vartheta_{ij} \sum \limits_{i''=1}^{|\mathcal{I}_j|} \frac{\vartheta_{i''j}}{\left(\vartheta_{ij}+ \vartheta_{i''j} \left(2^{-\frac{\mathrm{ln}(z)}{t}}-1\right)\right)}  \right)}{\prod \limits_{i'=1}^{|\mathcal{I}_j|} \left(\vartheta_{ij}+\vartheta_{i'j} \left(2^{-\frac{\mathrm{ln}(z)}{t}}-1\right) \right) } \cdot \frac{\mathrm{ln}(2) e^{-\frac{2^{-\frac{\mathrm{ln}(z)}{t}}-1}{\vartheta_{ij}}}}{t \; z \; 2^{\frac{\mathrm{ln}(z)}{t}}},\; \forall \; 0<z\leq 1. 
\end{equation}    
Due to the complexity of $Z_n$'s pdf, this expression cannot be used to calculate its mean. We propose to upper bound $f_{Z_n}^0$ by exploiting the following property: $0<z\leq 1$ implies that $\forall i'=1,\ldots,|\mathcal{I}_j|$, $\frac{1}{\vartheta_{ij}+\vartheta_{i'j} \left(2^{-\frac{\mathrm{ln}(z)}{t}}-1\right)}\leq \frac{1}{\vartheta_{ij}}$. Thus, $f_{Z_n}^0(z) \leq   f_{Z_n}(z)\cdot \frac{\left(1+\sum \limits_{i''=1}^{|\mathcal{I}_j|}\vartheta_{i''j}\right)}{\vartheta_{ij}^{|\mathcal{I}_j|-1}}$. By proceeding similarly as in (\ref{eq:develop4}), and combining the result with (\ref{eq:develop1}), $\zeta_1(T,\mathcal{I}_j)$ is obtained as presented in (\ref{eq:Chernoff2}).     
This completes the proof of Theorem \ref{Theorem0}.


\ifCLASSOPTIONcaptionsoff
  \newpage
\fi



%


\bibliographystyle{IEEEtran}
\bibliography{IEEEabrv,tau}

\begin{thebibliography}{10}
\providecommand{\url}[1]{#1}
\csname url@samestyle\endcsname
\providecommand{\newblock}{\relax}
\providecommand{\bibinfo}[2]{#2}
\providecommand{\BIBentrySTDinterwordspacing}{\spaceskip=0pt\relax}
\providecommand{\BIBentryALTinterwordstretchfactor}{4}
\providecommand{\BIBentryALTinterwordspacing}{\spaceskip=\fontdimen2\font plus
\BIBentryALTinterwordstretchfactor\fontdimen3\font minus
  \fontdimen4\font\relax}
\providecommand{\BIBforeignlanguage}[2]{{%
\expandafter\ifx\csname l@#1\endcsname\relax
\typeout{** WARNING: IEEEtran.bst: No hyphenation pattern has been}%
\typeout{** loaded for the language `#1'. Using the pattern for}%
\typeout{** the default language instead.}%
\else
\language=\csname l@#1\endcsname
\fi
#2}}
\providecommand{\BIBdecl}{\relax}
\BIBdecl

\bibitem{Wang2018}
C.~Wang, Y.~He, F.~R. Yu, Q.~Chen, and L.~Tang, ``{Integration of Networking,
  Caching, and Computing in Wireless Systems: A Survey, Some Research Issues,
  and Challenges},'' \emph{IEEE Commun. Surveys Tuto.}, vol.~20, no.~1, pp.
  7--38, $1^{st}$ quarter 2018.

\bibitem{Agiwal2016}
M.~Agiwal, A.~Roy, and N.~Saxena, ``{Next Generation 5G Wireless Networks: A
  Comprehensive Survey},'' \emph{IEEE Commun. Surveys Tuto.}, vol.~18, no.~3,
  pp. 1617--1655, $3^{rd}$ quarter 2016.

\bibitem{Cheng2016}
W.~Cheng, X.~Zhang, and H.~Zhang, ``{Statistical-QoS Driven Energy-Efficiency
  Optimization Over Green 5G Mobile Wireless Networks},'' \emph{IEEE J. Sel.
  Areas Commun.}, vol.~34, no.~12, pp. 3092--3107, Dec. 2016.

\bibitem{Ansari}
R.~I. Ansari, C.~Chrysostomou, S.~A. Hassan, M.~Guizani, S.~Mumtaz,
  J.~Rodriguez, and J.~J. P.~C. Rodrigues, ``{5G D2D Networks: Techniques,
  Challenges, and Future Prospects},'' \emph{IEEE Systems J.}, vol.~PP, no.~99,
  pp. 1--15, 2018.

\bibitem{Malan2015}
F.~Malandrino, Z.~Limani, C.~Casetti, and C.~Chiasserini, ``{Interference-Aware
  Downlink and Uplink Resource Allocation in HetNets With D2D Support},''
  \emph{IEEE Trans. Wireless Commun.}, vol.~14, no.~5, pp. 2729--2741, May
  2015.

\bibitem{Ali2016}
M.~Ali, S.~Qaisar, M.~Naeem, and S.~Mumtaz, ``{Energy Efficient Resource
  Allocation in D2D-Assisted Heterogeneous Networks with Relays},'' \emph{IEEE
  Access}, vol.~4, pp. 4902--4911, 2016.

\bibitem{Huang2016}
Y.~Huang, A.~A. Nasir, S.~Durrani, and X.~Zhou, ``{Mode Selection, Resource
  Allocation, and Power Control for D2D-Enabled Two-Tier Cellular Network},''
  \emph{IEEE Trans. Commun.}, vol.~64, no.~8, pp. 3534--3547, Aug. 2016.

\bibitem{Cao2017}
W.~Cao, G.~Feng, S.~Qin, and M.~Yan, ``{Cellular Offloading in Heterogeneous
  Mobile Networks With D2D Communication Assistance},'' \emph{IEEE Trans. Veh.
  Tech.}, vol.~66, no.~5, pp. 4245--4255, May 2017.

\bibitem{Tsiropoulos2017}
G.~I. Tsiropoulos, A.~Yadav, M.~Zeng, and O.~A. Dobre, ``{Cooperation in 5G
  HetNets: Advanced Spectrum Access and D2D Assisted Communications},''
  \emph{IEEE Wireless Commun.}, vol.~24, no.~5, pp. 110--117, Oct. 2017.

\bibitem{Naqvi2018}
S.~A.~R. Naqvi, H.~Pervaiz, S.~A. Hassan, L.~Musavian, Q.~Ni, M.~A. Imran,
  X.~Ge, and R.~Tafazolli, ``{Energy-Aware Radio Resource Management in
  D2D-Enabled Multi-Tier HetNets},'' \emph{IEEE Access}, vol.~6, pp.
  16\,610--16\,622, 2018.

\bibitem{Hao2018}
Y.~Hao, Q.~Ni, H.~Li, and S.~Hou, ``{Robust Multi-Objective Optimization for
  EE-SE Tradeoff in D2D Communications Underlaying Heterogeneous Networks},''
  \emph{IEEE Trans. Commun.}, vol.~66, no.~10, pp. 4936--4949, Oct. 2018.

\bibitem{Bastug2014}
E.~Bastug, M.~Bennis, and M.~Debbah, ``{Living on the Edge: The Role of
  Proactive Caching in 5G Wireless Networks},'' \emph{IEEE Commun. Mag.},
  vol.~52, no.~8, pp. 82--89, Aug 2014.

\bibitem{Li2018}
L.~Li, G.~Zhao, and R.~S. Blum, ``{A Survey of Caching Techniques in Cellular
  Networks: Research Issues and Challenges in Content Placement and Delivery
  Strategies},'' \emph{IEEE Commun. Surveys Tuto.}, vol.~20, no.~3, pp.
  1710--1732, 2018.

\bibitem{Peng2015}
X.~Peng, J.~C. Shen, J.~Zhang, and K.~B. Letaief, ``{Backhaul-Aware Caching
  Placement for Wireless Networks},'' in \emph{Proc. IEEE Global Commun. Conf.
  (GLOBECOM)}, Dec. 2015, pp. 1--6.

\bibitem{Blasz2015}
B.~Blaszczyszyn and A.~Giovanidis, ``Optimal geographic caching in cellular
  networks,'' in \emph{Proc. IEEE Int. Conf. Commun. (ICC)}, Jun. 2015, pp.
  3358--3363.

\bibitem{Kreishah2015}
A.~Khreishah and J.~Chakareski, ``{Collaborative caching for
  multicell-coordinated systems},'' in \emph{Proc. IEEE Conf. Comput. Commun.
  Wrkshp. (INFOCOM WKSHPS)}, Apr. 2015, pp. 257--262.

\bibitem{Shanmugam2013}
K.~Shanmugam, N.~Golrezaei, A.~G. Dimakis, A.~F. Molisch, and G.~Caire,
  ``{FemtoCaching: Wireless Content Delivery Through Distributed Caching
  Helpers},'' \emph{IEEE Trans. Inf. Theory}, vol.~59, no.~12, pp. 8402--8413,
  Dec. 2013.

\bibitem{Guan2014}
Y.~Guan, Y.~Xiao, H.~Feng, C.~Shen, and L.~J. Cimini, ``Mobicacher:
  Mobility-aware content caching in small-cell networks,'' in \emph{Proc. IEEE
  Global Commun. Conf.}, Dec. 2014, pp. 4537--4542.

\bibitem{Yang2014}
C.~Yang, Z.~Chen, Y.~Yao, B.~Xia, and H.~Liu, ``{Energy efficiency in wireless
  cooperative caching networks},'' in \emph{Proc. IEEE Int. Conf. Commun.
  (ICC)}, Jun. 2014, pp. 4975--4980.

\bibitem{Asadi}
A.~Asadi, Q.~Wang, and V.~Mancuso, ``{A Survey on Device-to-Device
  Communication in Cellular Networks},'' \emph{IEEE Commun. Surveys Tuto.},
  vol.~16, no.~4, pp. 1801--1819, $4^{th}$ quarter 2014.

\bibitem{Ji2016}
M.~Ji, G.~Caire, and A.~F. Molisch, ``{Wireless Device-to-Device Caching
  Networks: Basic Principles and System Performance},'' \emph{IEEE J. Sel.
  Areas Commun.}, vol.~34, no.~1, pp. 176--189, Jan 2016.

\bibitem{Zhang2016}
L.~Zhang, M.~Xiao, G.~Wu, and S.~Li, ``{Efficient Scheduling and Power
  Allocation for D2D-Assisted Wireless Caching Networks},'' \emph{IEEE Trans.
  Commun.}, vol.~64, no.~6, pp. 2438--2452, Jun. 2016.

\bibitem{Chen2017}
B.~Chen, C.~Yang, and Z.~Xiong, ``{Optimal Caching and Scheduling for
  Cache-Enabled D2D Communications},'' \emph{IEEE Commun. Letters}, vol.~21,
  no.~5, pp. 1155--1158, May 2017.

\bibitem{Yi2018}
C.~Yi, S.~Huang, and J.~Cai, ``{An Incentive Mechanism Integrating Joint Power,
  Channel and Link Management for Social-Aware D2D Content Sharing and
  Proactive Caching},'' \emph{IEEE Trans. Mob. Comp.}, vol.~17, no.~4, pp.
  789--802, Apr. 2018.

\bibitem{Yang2016}
C.~Yang, Y.~Yao, Z.~Chen, and B.~Xia, ``{Analysis on Cache-Enabled Wireless
  Heterogeneous Networks},'' \emph{IEEE Trans. Wireless Commun.}, vol.~15,
  no.~1, pp. 131--145, Jan. 2016.

\bibitem{Li2017}
Y.~Li, M.~C. Gursoy, and S.~Velipasalar, ``{A Delay-Aware Caching Algorithm for
  Wireless D2D Caching Networks},'' in \emph{Proc. IEEE Conf. Compu. Commun.
  Worksh. (INFOCOM WKSHPS)}, May 2017, pp. 456--461.

\bibitem{Li2018_2}
K.~Li, C.~Yang, Z.~Chen, and M.~Tao, ``{Optimization and Analysis of
  Probabilistic Caching in $N$ -Tier Heterogeneous Networks},'' \emph{IEEE
  Trans. Wireless Commun.}, vol.~17, no.~2, pp. 1283--1297, Feb. 2018.

\bibitem{Yang2018}
Z.~Yang, C.~Pan, Y.~Pan, Y.~Wu, W.~Xu, M.~Shikh-Bahaei, and M.~Chen, ``{Cache
  Placement in Two-Tier HetNets with Limited Storage Capacity: Cache or
  Buffer?}'' \emph{IEEE Trans. Commun.}, pp. 1--1, 2018.

\bibitem{Quer2018}
G.~Quer, I.~Pappalardo, B.~D. Rao, and M.~Zorzi, ``{Proactive Caching
  Strategies in Heterogeneous Networks With Device-to-Device Communications},''
  \emph{IEEE Trans. Wireless Commun.}, vol.~17, no.~8, pp. 5270--5281, Aug.
  2018.

\bibitem{Jaafar2018}
W.~Jaafar, W.~Ajib, and H.~Elbiaze, ``{Joint Caching and Resource Allocation in
  D2D-Assisted Heterogeneous Networks},'' in \emph{Proc. 14th Int. Conf. on
  Wireless and Mob. Comput., Network. and Commun. (WiMob)}, Oct. 2018, pp.
  1--8.

\bibitem{Amer2018}
R.~Amer, M.~M. Butt, M.~Bennis, and N.~Marchetti, ``{Inter-Cluster Cooperation
  for Wireless D2D Caching Networks},'' \emph{IEEE Trans. Wireless Commun.},
  vol.~17, no.~9, pp. 6108--6121, Sep. 2018.

\bibitem{Crama2011}
Y.~Crama and P.~L. Hammer, \emph{{Boolean Functions: Theory, Algorithms, and
  Applications}}.\hskip 1em plus 0.5em minus 0.4em\relax Cambridge Univ. Press,
  2011.

\bibitem{AMPL}
\BIBentryALTinterwordspacing
(2016) {AMPL: S}treamlined modeling for real optimization. [Online]. Available:
  \url{http://ampl.com}
\BIBentrySTDinterwordspacing

\bibitem{Cplex}
\BIBentryALTinterwordspacing
(2016) {IBM CPLEX O}ptimizer. [Online]. Available:
  \url{http://www-01.ibm.com/software/commerce/ \\
  optimization/cplex-optimizer/}
\BIBentrySTDinterwordspacing

\bibitem{mitzenmacher_upfal_2005}
M.~Mitzenmacher and E.~Upfal, \emph{{Probability and Computing: Randomized
  Algorithms and Probabilistic Analysis}}.\hskip 1em plus 0.5em minus
  0.4em\relax Cambridge University Press, 2005.

\end{thebibliography}

%








\end{document}